\documentclass{article}
\usepackage{arxiv}

\usepackage[table]{xcolor}
\usepackage{booktabs} 

\usepackage{xspace}

\usepackage[nice]{nicefrac}
\usepackage{algpseudocode}
\usepackage[small]{caption}
\usepackage{amssymb}
\usepackage[round]{natbib}
\usepackage{mathtools}
\usepackage{amsmath}
\usepackage{tikz}
\usetikzlibrary{calc}
\usepackage{todonotes}
\usepackage{dsfont}
\usepackage{enumitem}
\usepackage{amsthm}
\usepackage{thm-restate}
\usepackage{booktabs}
\usepackage{multirow}
\usepackage{newfloat}
\usepackage{wrapfig}
\usepackage{dsfont}
\usepackage{comment}
\usepackage{multirow}
\usetikzlibrary{decorations.markings}
\usepackage{cleveref}

\usepackage{pifont}
\usepackage{array}

\newtheorem{example}{Example}
\newtheorem{theorem}{Theorem}
\newtheorem{lemma}{Lemma}
\newtheorem{remark}{Remark}
\newtheorem{corollary}{Corollary}
\newtheorem{proposition}{Proposition}
\newtheorem{definition}{Definition}
\newtheorem{observation}{Observation}
\newtheorem{assumption}{Assumption}
\newtheorem{condition}[theorem]{Condition}

\newcommand{\nb}[3]{{\colorbox{#2}{\bfseries\sffamily\scriptsize\textcolor{white}{#1}}}{\textcolor{#2}{\sf\small\textit{#3}}}}

\newcommand{\anna}[1]{\nb{Anna}{green}{#1}}

\newcommand{\mat}[1]{\nb{Matteo}{blue}{#1}}

\newcommand{\BigOL}[1]{\tilde{\mathcal{O}}\left(#1\right)}

\newcommand{\bvec}{\mathbf{b}}
\newcommand{\rev}{\textsf{Rev}}
\newcommand{\gft}{\textsf{GFT}}
\newcommand{\opt}{\textsf{Opt}}
\newcommand{\GBB}{\textsf{GBB}\xspace}
\newcommand{\SBB}{\textsf{SBB}\xspace}
\newcommand{\WBB}{\textsf{WBB}\xspace}

\DeclareMathOperator*{\argmax}{arg\,max}

\DeclareMathOperator*{\argmin}{arg\,min}

\usepackage{algorithm}

\usepackage[utf8]{inputenc} 
\usepackage[T1]{fontenc}    
\usepackage{url}            
\usepackage{booktabs}       
\usepackage{amsfonts}       
\usepackage{nicefrac}       
\usepackage{microtype}      
\usepackage{lipsum}

\author{
	 Anna Lunghi\\
	Politecnico di Milano\\
	\texttt{anna.lunghi@polimi.it}
	\And
	 Matteo Castiglioni\\
	Politecnico di Milano\\
	\texttt{matteo.castiglioni@polimi.it}
		\And
	 Alberto Marchesi\\  
	Politecnico di Milano\\
	\texttt{alberto.marchesi@polimi.it}}

\title{Online Two-Sided Markets: \\Many Buyers Enhance Learning}

\begin{document}
\maketitle
\begin{abstract}
   We study a repeated trading problem in which a mechanism designer facilitates trade between a \emph{single seller} and \emph{multiple buyers}.
   Our model generalizes the classic bilateral trade setting to a multi-buyer environment.
   Specifically, the mechanism designer runs a second-price auction among the buyers—extending the fixed-price mechanism used in bilateral trade—before proposing a price to the seller.
   While this setting introduces new challenges compared to bilateral trade, it also provides an \emph{informational advantage}.
   Indeed, the presence of multiple buyers enhances competition, inducing them to reveal their valuations in order to win the auction.
   However, as in bilateral trade, the seller faces a binary decision: whether to accept the proposed price or not.
   
    We show that this \emph{asymmetric feedback}, which is more informative than in bilateral trade, allows us to break some lower bounds on regret minimization with a single buyer.
    In particular, we provide a $\tilde O(T^{2/3})$ regret upper bound with respect to an optimal \emph{strong budget-balanced} mechanism, without any assumptions on the distribution of valuations.
    Our main tool for achieving this result is the design of an adaptive grid that approximates the optimal gain from trade across the continuum of possible mechanisms.
    Furthermore, we attain the same regret bound with respect to an optimal \emph{global budget-balanced mechanism}, under two possible conditions: (i) buyers' and seller’s valuations are independent, or (ii) valuations are drawn from a distribution with bounded density.
    In doing so, we provide some novel technical results on constrained MABs with feedback graphs, which may be of independent interest.
\end{abstract}

\section{Introduction}

Facilitating trades between sellers and buyers is a fundamental aspect of economic systems. In this work, we focus on two-sided markets with a \emph{single} seller and \emph{multiple} buyers, where each agent seeks to trade a single item. Such settings are typically modeled as a mechanism design problem, where agents have private valuations for the item and aim to maximize their own utilities.

Ideally, the mechanism should ensure 
\emph{efficiency}, which means maximizing the sum of utilities. A standard tool to address this problem is the VCG mechanism~\citep{nisan2007algorithmic}. However, even in simple bilateral trade scenarios, the VCG mechanism fails to achieve \emph{budget balance}, requiring the mechanism designer to subsidize the market and thus incur financial losses. A seminal result by \citet{myerson1983efficient} shows that no fully efficient mechanism can simultaneously satisfy incentive compatibility, individual rationality, and budget balance.

A recent line of work initiated by \citet{cesa2023repeated} tries to circumvent this impossibility result through the lens of regret minimization. In this setting, at each round $t$, a new seller and a set of $n$ new buyers arrive, with their valuations drawn from an unknown (possibly joint) distribution.
The mechanism first sets a reserve price $q_t$ for the buyers.
If the highest buyer valuation $\overline{b}_t\ge q_t$, then the corresponding buyer is willing to trade. Similarly, a price $p_t$ is set for the seller, who is willing to trade if their valuation $s_t\le p_t$. If both parties agree the resulting \emph{gain from trade} is
\[\gft_t(p_t,q_t) \coloneqq (\overline{b}_t-s_t)\mathbb{I}(s_t\le p_t) \mathbb{I}(  \overline{b}_t \ge q_t) .\]

The learner's objective is to maximize the cumulative gain from trade. Unlike aiming for full efficiency, this approach relaxes the baseline and focuses on the less demanding task of achieving performance comparable to the best fixed-price mechanism.
In doing so, the mechanism must ensure some budget balance constraints.
Such constraints guarantee that the mechanism designer is \emph{not} subsidizing the market and are thus deeply connected with its revenue.
When a trade occurs, the buyer pays the maximum between $q_t$ and the second-highest buyer valuation $\underline{b}_t$ (as in a second-price auction with reserve price), while the seller receives $p_t$. Thus, the mechanism revenue is
\[\rev_{t}(p_t,q_t)= (\max\{q_t, \underline{b}_t\}-p_t)\mathbb{I}(s_t\le p_t) \mathbb{I}(\overline b_t \geq q_t)    .\]

Previous works focus on bilateral trade problems \citep{cesa2023repeated,cesa2024bilateral,bernasconi2024no,azar2022alpha}, which involve a single seller and a \emph{single} buyer.
In this setting, the primary focus is on the realistic two-bit feedback model, where the learner only observes $\mathbb{I}(s_t\le p_t)$ and $ \mathbb{I}( \overline b_t \geq q_t)$ after each round $t$. Intuitively, agents only reveal whether they accept the proposed prices or not.
Early works in this area impose an ``hard'' budget balance constraint ensuring that either $\rev_t(p_t,q_t)\ge 0$, \emph{i.e.}, Weak Budget Balance (\WBB), or $\rev_t(p_t,q_t)= 0$, \emph{i.e.}, Strong Budget Balance (\SBB), at every round $t$.
In general, optimal mechanisms are unlearnable.
Therefore, \citet{cesa2024bilateral} provide some positive results under certain assumptions on the underlying valuations distribution.
Specifically, they provide a tight $\BigOL{T^{2/3}}$ regret bound when valuations are drawn independently among the seller and the buyers, and i.i.d. from a distribution with a bounded density.
Moreover, they provide a $\BigOL{T^{3/4}}$ regret bound when only the second assumption holds, which \citet{cesa2023repeated} show to be tight.

\citet{bernasconi2024no} show that it is possible to obtain positive results without any assumption on the underlying distribution (and even in adversarial settings).
To do so, they relax the budget balance constraint, requiring it to hold globally.
In particular, under a Global Budget Balance (\GBB) constraint, the mechanism designer only has to ensure that $\sum_t \rev_t(p_t,q_t)\ge 0$.
Note that despite using \GBB~mechanisms, \citet{bernasconi2024no} still compare their mechanisms to an optimal \SBB~one.
A detailed summary of previous results is presented in Table~\ref{tab:related}.

To the best of our knowledge, we are the first to consider a setting with multiple buyers (and a single seller) in an online learning scenario.
Despite our framework is a generalization of bilateral trade in many aspects, our work is based on the observation that the presence of multiple buyer provides an informational advantage, which enhances learning. Indeed, while in bilateral trade the buyer simply reveals whether they want to accept the price or not, the presence of multiple buyers increases competition: buyers have to reveal their valuations in order to win the second-price auction.
However, the seller still only decides whether accept or reject the proposed price. Hence, our feedback is still much less powerful than full feedback.

In this paper, we show that such \emph{asymmetric} feedback can be exploited to circumvent impossibility results in bilateral trade (with two-bit feedback).
In particular, we provide a $\BigOL{T^{2/3}}$ regrte bound for $\SBB$~settings without any assumption on the underlying valuations distribution.
Moreover, we provide the first no-regret guarantees with respect to \GBB baselines. To the best of our knowledge, these are the first sublinear regret guarantees with respect to \GBB mechanisms.

\subsection{Our Results and Techniques}

Our first result is an assumption-free $\BigOL{T^{2/3}}$ regret upper bound for \SBB~mechanisms. Here, the main technical challenge is to circumvent the $\Omega(T)$ lower bound of \citet{cesa2024bilateral}. This lower bound is essentially proven by showing that the mechanism must find a price that is a ``needle in a haystack''. This prevents the use of (non-adaptive) discretization techniques commonly used in continuous scenarios.
We show that our asymmetric feedback can be used to design an \emph{adaptive} grid to work on. Interestingly, we show that it is sufficient to use an adaptive grid only on the buyers' price, on which we have a stronger feedback.
Our main idea is to build a finer grid on the regions of the decision space that exhibit higher probability mass.
Conversely, we use a coarser grid in regions where the probability density is lower, minimizing unnecessary exploration in less relevant areas.
Intuitively, we want to include in the grid the prices corresponding to valuations that appear with sufficiently high probability, while guaranteeing that valuations in between two points of the grid occurs with small probability.
Such property mimics the bounded density assumption that in previous works has been essentially used to bound the probability that a valuation lies between two points of an uniform grid~\citep{cesa2023repeated}.
Then, equipped with our adaptive grid of size $K=\BigOL{T^{1/3}}$, we implement a regret minimizer that achieves regret $\BigOL{T^{2/3}}$. 
We remark that finding optimal \SBB mechanisms is essentially a single-dimensional problem, since the buyers' price $q_t$ determines the seller price $\max\{q_t,\underline{b}_t\}$.
As we discuss in the following, this is \emph{not} the case for \GBB~mechanisms, which are two-dimensional.

Our second main result  establishes the first sublinear regret guarantees with respect to \GBB mechanisms. In bilateral trade scenarios, deterministic \GBB mechanisms are effectively equivalent to \WBB ones since the only way to achieve positive expected revenue is by posting prices $p_t\le q_t$. This equivalence no longer holds in our two-sided market setting. Specifically, the mechanism gains a surplus of $\max\{q_t,\underline{b}_t\}-q_t$.
The only work that considers a similar baseline (but with randomization) in bilateral trade is \citep{bernasconi2024no}. However, their focus is on an adversarial setting, where they show that linear regret is unavoidable.
This impossibility result is based on the observation that learning optimal \GBB mechanisms incorporates an online learning problem with unknown constraints. 
In such settings, we cannot simply run a regret minimizer over a grid of possible prices, since we do \emph{not} know which actions are feasible \emph{a priori}. We show that this challenge can be overcome in a stochastic environment.
As an additional challenge, we have a huge set of arms, since if each price lies on a grid of size $K$, we have to optimize over $K^2$ couples, which makes impossible to learn without sharing the feedback.
Our proof relies on three main components: (i) an adaptive grid that works only when the buyers and seller distributions are independent or the joint distribution has bounded density, (ii) a multiplicative-additive revenue maximizing grid, and (iii) a tool to share feedback among arms in constrained MABs.

Our adaptive grid is effective only under specific assumptions on the valuation distribution. We show that these assumptions are required to learn an optimal \emph{fixed} couple of prices. However, this does \emph{not} rule out the possibility of achieving no regret through more sophisticated regret-minimization algorithms. We leave as an open problem to determine whether sublinear regret is attainable using different approaches.
The second component of our framework is a multiplicative-additive revenue-maximizing grid. This grid will be used to compensate a slight violation of the budget balance constraint due to exploration and our estimation errors.
Finally, we introduce a novel problem related to constrained MABs with feedback graphs. We restrict to graphs that satisfy two main properties: they consist of $n$ set of nodes (represented by indexes), in which each node in a set provides feedback on all the nodes in the same set with higher index, and (ii) constraints are monotone, \emph{i.e.}, if an arm $B$ is unfeasible and there is an edge from $A$ to $B$, then $A$ is also unfeasible.
To tackle this problem we design a novel algorithm based on successive arm elimination \emph{with precedences}, which provides a $\tilde{O}({\sqrt{T n}})$ regret bound (where we omitted a logarithmic dependence from the number of arms) and $\tilde{O}({\sqrt{T n}})$ positive constraint violations.
By combining these results, we provide a $\BigOL{T^{2/3}}$ regret bound when either the buyers and seller valuations are independent or their joint distribution admits bounded density.

\begin{table}
\caption{Comparison with prior works in the literature. ``\textbf{ind.}'' denotes that the distribution of buyer(s) and seller valuations are assumed independent, while ``\textbf{b.d.}'' denotes that the joint distribution has bounded density. Notice that the column ``budget balance'' refers to the baseline, while the mechanisms used are denoted by the symbols beside: results marked with $\star$ consider \SBB mechanisms, the ones marked with $\dagger$ consider \WBB mechanisms, and the ones marked with $\ddagger$ consider \GBB mechanisms.}
\begin{center}
 \renewcommand*{\arraystretch}{1.5}{\begin{tabular}{ l | c | c | c | c | c  } 
  \toprule \label{tab:related}
    & \textbf{\# buyers} & \textbf{ind.} & \textbf{b.d.}  & \textbf{budget balance}   & \textbf{regret} \\ 
   \midrule
 \multirow{3}{*}{\citet{cesa2024bilateral}}& \multirow{5}{*}{1}   & \checkmark & \ding{55} & strong$^\star$& $\Omega(T)$ \\ 
  \cline{3-6}
 &   & \checkmark & \checkmark & strong$^\star$ & $\Omega(T^{2/3})$ --- $\BigOL{T^{2/3}}$\\ 
  \cline{3-6}
 &   & \ding{55} & \checkmark & strong$^\dagger$ & $\BigOL{T^{3/4}}$\\  
 \cline{1-1} \cline{3-6}
 \citet{cesa2023repeated}&   & \ding{55} & \checkmark & strong$^\dagger$ & $\Omega(T^{3/4})$\\  
  \cline{1-1} \cline{3-6}
  \citet{bernasconi2024no}&  &  \ding{55} &  \ding{55} & strong$^\ddagger$ & $\Omega(T^{5/7})$ --- $\BigOL{T^{3/4}}$\\ 
\midrule
 \cellcolor{gray!20} & \cellcolor{gray!20}  & \cellcolor{gray!20}\ding{55} & \cellcolor{gray!20}\ding{55} & \cellcolor{gray!20}strong$^\dagger$ & \cellcolor{gray!20}$\Omega(T^{2/3})$ --- $\BigOL{T^{2/3}}$\\ 
  \cline{3-6}
 \cellcolor{gray!20}& \cellcolor{gray!20}  & \cellcolor{gray!20}\ding{55} & \cellcolor{gray!20}\checkmark  & \cellcolor{gray!20}global$^\ddagger$ & \cellcolor{gray!20}$\Omega(T^{2/3})$ --- $\BigOL{T^{2/3}}$\\ 
 \cline{3-6}
  \cellcolor{gray!20}&  \cellcolor{gray!20}  & \cellcolor{gray!20}\checkmark & \cellcolor{gray!20}\ding{55}  & \cellcolor{gray!20}global$^\ddagger$ & \cellcolor{gray!20}$\Omega(T^{2/3})$ --- $\BigOL{T^{2/3}}$\\ 
 \cline{3-6}
   \cellcolor{gray!20}\multirow{-4}{*}{\bf This paper}& \cellcolor{gray!20}\multirow{-4}{*}{$\geq 2$} & \cellcolor{gray!20}\ding{55} & \cellcolor{gray!20}\ding{55}  & \cellcolor{gray!20}global$^\ddagger$ & \cellcolor{gray!20} $\Omega(T^{2/3})$ --- open \\ 
  \bottomrule
\end{tabular}}
\end{center}
\end{table}
\subsection{Further Related Works}

\paragraph{Bayesian mechanism design in two-sided markets}
Given the impossibility result by \citet{myerson1983efficient}, many works try to design approximation algorithms in Bayesian settings.
There exist a large class of incentive compatible mechanisms that provide a constant-factor approximation to the social welfare~\citep{blumrosen2014reallocation,kang2022fixed}) or the gain from trade~\citep{mcafee2008gains,blumrosen2016approximating,brustle2017approximating,deng2022approximately,fei2022improved}.

\paragraph{Online learning in bilateral trade}
Table~\ref{tab:related} provides a detailed summary of previous work, with a particular focus on stochastic settings. We mention that a line of works has investigated regret minimization in adversarial settings. \citet{azar2022alpha} study weak budget balanced mechanisms and show that sublinear 2-regret when the valuations are adversarial, which is tight.
\citet{cesa2023repeated} prove that sublinear regret can be attained against a smoothed adversary. Finally, \cite{bernasconi2024no} show that sublinear regret is achievable without any assumptions on the adversary using \GBB mechanisms.
Another recent related work is \citep{babaioff2024learning}, which studies the sample complexity of learning optimal mechanisms in two-sided markets with a single seller and two buyers.

\paragraph{Bandits with knapsack (and more general) constraints} 
Another line of research relevant for two-sided markets with \emph{global} budget constraints are bandits with knapsacks and, more generally, bandits with long-term constraints.
Indeed, similarly to long-term constraint, the constraint on the revenue must holds globally over all the rounds.  Here, we focus on works with \emph{stochastic} constraints, which are the most relevant to our work.
The Bandit with knapsack model is introduced by \citet{badanidiyuru2018bandits}, while  \cite{kesselheim2020online} improve the dependence from the number of resources. Subsequent works extend the framework in many directions like non-monotone resource utilization~\cite{kumar2022non,bernasconibandits} and contextual bandits~\cite{bernasconi2024no,slivkins2023contextual}.
Some of these works (see, \emph{e.g.}, \cite{slivkins2023contextual,bernasconi2024no,castiglioni2022unifying,DBLP:conf/icml/CastiglioniCK24}) provide a general framework for addressing bandit with knapsack assuming the existence of two suitable primal and dual regret minimizer.
In principle, this approach could be applied to our constrained MAB with graph feedback by leveraging a black-box regret minimizer for standard (unconstrained) MABs with graph feedback~\citep{alon2015online}.
However, these methods heavily rely on the Slater condition, which is essential for achieving optimal regret bounds.
Since this condition does not hold in our setting (as we lack a lower bound on revenue), these approaches are ineffective in our setting.
As a result, we must develop novel techniques.

\paragraph{MABs with feedback graphs}
Another line of research related to ours is the study of MABs with graph feedback, which focuses on feedbacks different from the full or bandit feedback.
MABs with Feedback graphs are first introduced for adversarial bandits by \citet{mannor2011bandits}, while \citet{alon2017nonstochastic} improve the dependence of the regret bound to $\BigOL{\sqrt{T\alpha\log(K)}}$ regret bound, where $\alpha$ is the independence number of the graph and $K$ the number of arms. Adversarial MAB with feedback graphs have also been  studied in relation to different aspects of the problem like the observability of the graph,  the robustness to noise and the data-dependent guarantees (see, \emph{e.g.} \cite{alon2015online,Kocák2014,Kock2016OnlineLW,lykouris2018small,cohen2016online,lee20asmallloss}).
\citet{caron2012leveraging} are the first to consider stochastic bandits with feedback graphs. In this setting, \citet{lykouris2020feedback} analyze the performance of active arm elimination, UCB and Thomson sampling techniques proving a worst-case regret of $\BigOL{\sqrt{\alpha T\log(K)}}$.
All these works focus only on \emph{unconstrained} problems.

\section{Repeaded Two-Sided Markets With Many Buyers and a Single Seller}

We study online learning in a \emph{two-sided} trading problem, where \emph{many buyers} interact repeatedly with a \emph{single seller} over $T \in \mathbb{N}$ rounds. 
%
%
We focus on a \emph{stochastic} setting where the valuations of both the buyers and the seller are jointly drawn from a fixed distribution.
At each round $t\in [T]$,\footnote{In this paper, we use $[x]$ to denote the set $\{1,\ldots,x\}$ of the first $x \in \mathbb{N}$ natural numbers.} the market consists of $n\in \mathbb{N}$ buyers that are potentially interested in buying an item from a single seller.
The buyers are characterized by some private valuations $\bvec_t=(b_{t,1},\ldots,b_{t,n}) \in [0,1]^n$ for the item, while the seller has their own private valuation $s_t \in [0,1]$.
The buyers' and seller's valuations are sampled from a joint probability distribution $\mathcal{P}$ supported on $[0,1]^{n+1}$. Specifically, at each round $t \in [T]$, the pair $(\bvec_t,s_t)\sim \mathcal{P}$ is drawn independently from $\mathcal{P}$.

Our goal is to design mechanisms that maximize the gain from trade, so as to \emph{optimize the efficiency of the market}.
As customary in the related literature on bilateral trade, we focus on \emph{fixed-price} mechanisms, by considering a straightforward generalization of the mechanism commonly adopted in bilateral trade (see, \emph{e.g.},~\citep{cesa2023repeated,cesa2024bilateral,bernasconi2024no}).
At each round $t \in [T]$, the mechanism first implements a second-price auction among the buyers, by setting a reserve price $q_{t} \in [0,1]$. Let $\overline b_t \coloneqq \max_{b_{t,i} \in \bvec_t} b_{t,i}$ and $\underline b_t \coloneqq \max_{b_{t,i} \in \bvec_t : b_{t,i} \neq \overline b_t} b_{t,i}$ be the highest and the second-highest buyer valuations at round $t$, respectively, assuming ties are broken arbitrarily. If the highest buyer valuation exceeds the reserve price, \emph{i.e.}, $\overline b_t\ge q_{t}$, the highest-valuation buyer wins the possibility of buying the item at a price equal to $\max\{q_{t},\underline b_t \}$. Otherwise, no buyer wins the item. Then, the mechanism proposes a selling price $p_{t} \in [0,1]$ to the seller, who accepts to sell the item only if $s_t\le p_{t}$. Finally, if the buyer with the highest valuation wins the item and the seller accepts to sell it, then the trade is completed and the resulting \emph{gain from trade} is
\[
    \gft_t(p_{t},q_{t}) \coloneqq ( \overline b_t-s_t )\mathbb{I}(s_t\le p_{t})\mathbb{I}(\overline{b}_t\ge q_{t}),
\]
where $\mathbb{I}(\cdot)$ denotes the indicator function.
%
%
In the following, we also use 
\[\gft(p,q)\coloneqq \mathbb{E}_{(\bvec,s)\sim \mathcal{P}} \left[ (\overline{b}-s)\mathbb{I}(s\le p)\mathbb{I}(\overline{b}\ge q) \right] \]
to denote the \emph{expected} gain from trade of a mechanism defined by a pair of fixed prices $(p,q) \in [0,1]^2$, where we let $\overline{b} \coloneqq \max_{b_i \in \bvec} b_i$ be the highest buyer valuation in the tuple $\bvec = (b_1, \ldots, b_n) \in [0,1]^n$ sampled from $\mathcal{P}$.
Beside being simple, it is easy to see that our mechanism enjoys the desirable property of being \emph{dominant-strategy incentive compatible} (DSIC) for both the buyers and the seller.

\subsection{Budget Balance}

While the primary objective in trading problems is to maximize the gain from trade, mechanisms typically need to satisfy an additional \emph{budget balance} property, which could be either \emph{weak}/\emph{strong} (see, \emph{e.g.},~\citep{cesa2023repeated,cesa2024bilateral}) or \emph{global} (see, \emph{e.g.},~\citep{bernasconi2024no}).
We consider two distinct classes of mechanisms based on the specific budget balance property being examined.
  

 \paragraph{ Per-round budget balance mechanisms: Weak and Strong budget balanced (\WBB~and \SBB)}
 
 Budget balance properties are concerned with the \emph{revenue} of the mechanism $(p_t,q_t) \in [0,1]^2$ implemented at round $t \in [T]$, which is defined as  
 \[ \rev_t(p_{t},q_{t}) \coloneqq (\max\{q_{t}, \underline b_t\}-p_{t})\mathbb{I}(s_t\le p_{t})\mathbb{I}(\overline{b}_t\ge q_{t}), \]
 where $\max\{q_{t}, \underline b_t\}-p_{t}$ represents the profit earned by the mechanism in the event that the trade occurs.
 As for the gain from trade, in the following we denote by 
 \[\rev(p,q) \coloneqq \mathbb{E}_{(s,\bvec)\sim \mathcal{P}}\left[(\max\{q, \underline b\}-p)\mathbb{I}(s\le p)\mathbb{I}(\overline{b}\ge q)\right]\]
 the \emph{expected} revenue of a pair of prices $(p,q) \in [0,1]^2$, where we let $\underline{b} \coloneqq \max_{b_i \in \bvec : b_i \neq \overline b} b_i$.

Budget balance generally guarantees that the mechanism is \emph{not} subsidizing the market. In particular, \WBB~ensures that at each round the mechanism is \emph{not} loosing money. Formally:
  \[ \rev_t(p_{t},q_{t})\ge 0 \quad \forall t \in [T]. \]
As it will become clear in the following, \WBB mechanisns are useful in learning scenarios. However, their gain from trade is dominated by the one of \SBB mechanisms. 
Indeed, the \SBB property guarantees that the mechanism is \emph{not} gaining any profit from the market, and, thus, \SBB mechanisms are the proper tool to maximize the gain from trade. Formally, \SBB requires that 
  \[ \rev_t(p_{t},q_{t})=0 \quad \forall t \in [T]. \]
The \SBB property removes a degree of freedom from the mechanism, by forcing $p_{t}=\max\{q_{t}, \underline b_t\}$. Hence, we can identify a \SBB~mechanism by only specifying its reserve price $q_{t} \in [0,1]$.
Since \SBB mechanisms are both simpler and as powerful as \WBB ones, similar to previous work on bilateral trade (see, \emph{e.g.},~\cite{cesa2023repeated,cesa2024bilateral}), one of our goals will be to achieve performance close to that of an optimal \SBB mechanism through the use of \WBB mechanisms.
%

 \paragraph{Global Budget Balance (\textsf{GBB})}
 
 This is a looser requirement that can be applied to mechanisms.
 Specifically, it does \emph{not} force the revenue to be null at each round, but it only requires that the mechanism is \emph{not} subsiding the market (\emph{i.e.}, losing money) overall.
 Formally:
 \[ 
    \sum_{t \in [T]} \rev_t(p_{t},q_{t})\ge 0.
    \]
Thus, under the \GBB property, the mechanism can choose the prices $p_{t}$ with the only requirement that the budget balance constraint holds globally over the $T$ rounds.

\subsection{Learning Problems}
 
 \paragraph{Feedback}
 
 After posting the reserve price $q_{t}$, the learner observes the bids of all the buyers with valuation larger than $q_{t}$.
 Moreover, the seller observes if $s_t\ge p_{t}$.
 This feedback is equivalent to the two-bit feedback used in bilateral trade for the seller, but it is stronger for the buyer. Such asymmetry is motivated by the second-price auction with reserved price to which buyers participate. We refer to such feedback as \emph{asymmetric} feedback.
 Algorithm~\ref{alg: model} summarizes the interaction between the learner and the environment.
 
\begin{algorithm}[H]\caption{Mechanism Interaction}\label{alg: model}
	\begin{algorithmic}[1]
		\For{$t=1,\ldots,T$}
        \State
		The environment samples the private valuations of the buyers and the seller $(\bvec_t,s_t)\sim \mathcal{P}$
		\State The mechanism posts a reserve price $q_{t} \in [0,1]$
		\State The mechanism runs a second-price auction among the buyers, with reserve price fixed as $q_{t}$
		\State The mechanism observes the buyer valuations $\left\{ b_{t,i}  \in \bvec_t \mid b_{t,i} \ge q_{t} \right\}$
		 \State The mechanism proposes a selling price $p_{t} \in [0,1]$ to the seller
		 \State The mechanism observes the seller's decision $\mathbb{I}(s_t\le p_{t})$
		\If{$ \mathbb{I}(s_t\le p_{t})\mathbb{I}(\overline{b}_t\ge q_{t})=1$}
        \State Trade happens between the highest-valuation buyer and the seller
         \State The highest-valuation buyer pays $\max\{\underline{b}_t,q_t\}$ to the mechanism and the seller is payed $p_{t}$
        \Else
        \State The trade does not happen
        \EndIf
		\EndFor
	\end{algorithmic}
\end{algorithm}	
 
 \paragraph{Performance measure}
 The learner's goal is to maximize $\sum_{t \in [T]} \gft_t(p_{t},q_{t})$, \emph{i.e.}, the cumulative gain from trade, subject to a weak or global budget balance constraint.
 The baseline depends on the specific setting.
 
 For \WBB settings, we take as baseline the best fixed reserve price (recall that optimal \WBB and \SBB mechanisms achieve the same gain from trade).
To keep the notation more agile we will use $\gft(q)$ and $\gft_t(q_t)$ to denote $\gft(\max\{\underline{b}_t,q\},q)$ and $\gft_t(\max\{\underline{b}_t,q\},q)$, respectively.
Then, the regret is defined as: 
 
 \begin{equation*}
 	R_T \coloneq T \cdot \max_{q\in[0,1]}\gft(q)-\mathbb{E}\left[\sum_{t\in[T]}\gft(q_t)\right],
 \end{equation*}
 where $q_t$ is the reserve price decided by the learning algorithm at each episode $t$, and the expectation is with respect to the randomness of the learner.
 
For \GBB settings, we define our baseline as the fixed prices $(p,q)$ that maximizes the gain from trade and guarantees positive revenue. Formally:
\begin{subequations}\label{prog:baseline}
 \begin{align}
	\opt \coloneq \max_{p,q\in[0,1]} & \quad \gft(p,q) \quad \text{s.t.}\\
	&  \rev(p,q)\ge0.
\end{align}
\end{subequations}
 Then,
 \[R_T= T \cdot \opt -\mathbb{E}\left[\sum_{t\in[T]}\gft_t(p_t,q_t)\right], \]
 where $(p_t,q_t)$ are the prices decided by the learning algorithm at each episode $t$, and the expectation is with respect to the randomness of the learner.

\paragraph{Environments}

Most of our results hold for arbitrary distribution $\mathcal{P}$.
In some settings, problems with general distributions are non-learnable and we have to restrict to more structured distributions:

\begin{itemize}
	\item \emph{Independent:} the distribution of the buyers $\mathcal{P}^B$ and of the seller $\mathcal{P}^S$ and  are independent, and $\mathcal{P}:=\mathcal{P}^B\times \mathcal{P}^S$,
	\item \emph{Bounded density:} $\mathcal{P}$ admits a density bounded by $M$.
\end{itemize}

\paragraph{Relation to bilateral trade}
Our setting is a generalization of the bilateral trade to multiple buyers. Indeed, setting the number of buyer $n=1$, we recover the online bilateral trade problem studied by \citet{cesa2023repeated,cesa2024bilateral,bernasconi2024no}. 
However, works on bilateral trade commonly assume two(or even one)-bit feedback. Formally, the learner observes $\mathbb{I}(s_t\le p_{t})$ and $\mathbb{I}(\overline{b}_t\ge q_{t})=1$. While our feedback on the seller valuation is essentially equivalent, our feedback is more informative on the buyers valuations. Intuitively, buyers with high valuations are ``forced'' to disclose their valuation to win the auction. Such \emph{asymmetric} feedback is probably the main feature of our model, and the main reason why we will be able to obtain better regret bound with respect to bilateral trade settings.

\section{Decomposition of Gain from Trade}\label{sec:decomposition}

Our asymmetric feedback sits between two-bit feedback and full feedback in terms of the information it provides. While it offers more insight than two-bit feedback, it doesn’t give the complete picture that full feedback offers. This makes important to find strategies that make the best use of the information available.
Even though our feedback is more detailed than two-bit feedback, it is still not enough to directly observe the $\gft$ generated by a specific set of prices. Indeed, it does not even guarantee bandit feedback.
The main drawack of this feedback is that it does not provide direct access to $s_t$, \emph{i.e.}, the seller’s valuation, which is a key factor in estimating the gain from trade. Our first result breaks down the $\gft(\cdot,\cdot)$ function into simpler components that underline the feedback structure. A similar result for bilateral trade is presented in \citep{cesa2024bilateral}.

\begin{restatable}{lemma}{gftdecompose}
	\label{lemma: gft decomposition}
	Fix any couple of prices $(p,q)\in [0,1]^2$. Define the functions $\gft_1: [0,1]^2 \rightarrow [-1,1]$ and $\gft_2: [0,1]^2 \rightarrow [-1,1]$ as
	\begin{align*}
		\gft_1(p,q) \coloneq \mathbb{E}_{(\bvec,s)\sim \mathcal{P}}[(p-s)\mathbb{I}(\bar{b}\ge q)\mathbb{I}(s\le p)],
	\end{align*}
	and
	\begin{align*}
		\gft_2(p,q) \coloneq \mathbb{E}_{(\bvec,s)\sim \mathcal{P}}[(\bar{b}-p)\mathbb{I}(\bar{b}\ge q)\mathbb{I}(s\le p)].
	\end{align*}
	Then 
	\begin{align*}
		\gft(p,q)= \gft_1(p,q) + \gft_2(p,q).
	\end{align*}
	Moreover, let $\mathcal{U}$ be the uniform distribution over $[0,1]$. Then 
	\begin{align*}
		\gft_1(p,q) =  \mathbb{P}_{U\sim \mathcal{U},(\bvec,s)\sim \mathcal{P}}(s\le U\le p, \overline{b} \ge q).
	\end{align*}
		
\end{restatable} 

Estimating $\gft_1$ and $\gft_2$ is still an hard task.
$\gft_1$ cannot be easily estimated without observing $s_t$ (even observing the two-bit feedback of the indicator functions), while $\gft_2$ can be estimated more easily having access to $\overline b_t$. 
However, as we showed in the second part of \Cref{lemma: gft decomposition}, the component \(\gft_1(p, q) = \mathbb{P}(s_t \le U \le p, \bar{b}_t \ge q)\) allows us to gather full feedback, provided we play prices of the form \((U, 0)\) with $U$ sampled uniformly from $[0,1]$. 
This means that by observing \(\mathbb{I}(s_t \le U)\), we can simultaneously infer \(\mathbb{I}(s_t \le U \le p)\) for all \(p \in [0, 1]\). In addition, setting the reserve price for buyers to $q=0$ enables us to observe all the buyers' valuations. Consequently, this allows us to compute \(\mathbb{I}(\bar{b}_t \ge q)\) for all \(q \in [0, 1]\) at the same time. Thus, it is possible to fully estimate $\gft_1$ if we accept the drawback of ``wasting'' rounds for exploration.

In contrast, the component $GFT_2$ has a more informative feedback than bandit feedback but still less informative than full feedback. This is because it depends on the chosen prices $(p, q)$. 
It is relatively easy to see why playing a specific pair of prices $(p, q)$ allows us to observe a unbiased estimator of \(\gft_2(p, q)\). If \(\overline{b}_t\), \emph{i.e.} the highest bid of the buyers, is not observed, then it must be the case that \(\mathbb{I}(\bar{b}_t \ge q) = 0\). On the other hand, if \(\bar{b}_t\) is observed, we gain access to its exact value. 
Additionally, the feedback from \(\gft_2\) is stronger than bandit feedback because playing \((p, q)\) provides information about \(\gft_2(p, q')\) for all \(q' \ge q\). Specifically, when \(\bar{b}_t\) is not observed, the condition \(\mathbb{I}(\bar{b}_t \ge q) = 0\) implies that \(\mathbb{I}(\bar{b}_t \ge q') = 0\) for any \(q' \ge q\).  
Unfortunately, the feedback is still equivalent to bandit feedback when it comes to the seller's price. Setting the seller's price to \(p\) does not provide any insight into the value of \(\gft_2\) for other seller prices \(p' \in [0, 1]\).

We can prove a similar result when $p$ is fixed to $\max\{q,\underline{b}\}$. This result is useful in strong budget balanced settings.

\begin{restatable}{lemma}{gftdecomposeSBB}
	\label{cor: GFT decomposition}
	Fix any  price $q\in [0,1]$. Define the functions $\gft_1: [0,1] \rightarrow [-1,1]$ and $\gft_2: [0,1] \rightarrow [-1,1]$ as
	\begin{align*}
		\gft_1(q) \coloneq \mathbb{E}_{(\bvec,s)\sim \mathcal{P}}[(\max\{p,\underline{b}\}-s)\mathbb{I}(\bar{b}\ge q)\mathbb{I}(s\le \max\{q,\underline{b}\})],
	\end{align*}
	and
	\begin{align*}
		\gft_2(q) \coloneq \mathbb{E}_{(\bvec,s)\sim \mathcal{P}}[(\bar{b}-\max\{q,\underline{b}\})\mathbb{I}(\bar{b}\ge q)\mathbb{I}(s\le \max\{q,\underline{b}\})].
	\end{align*}
	Then 
	\begin{align*}
		\gft(q)= \gft_1(q) + \gft_2(q).
	\end{align*}
	Moreover, let $\mathcal{U}$ be the uniform distribution over $[0,1]$. Then 
	\begin{equation*}
		\gft_1(q) =  \mathbb{P}_{U\sim \mathcal{U},(\bvec,s)\sim \mathcal{P}}(s_t\le U\le \max\{\underline{b}_t,q\}, \bar{b}_t \ge q).
	\end{equation*}
\end{restatable}

\section{From Continuous to Discrete Prices}\label{sec:grid}

In this section, we deal with the challenge of dealing on a continuous action set of possible prices $[0,1]^2$. As it is standard in such settings, we would like to work with a discretization of the action set including finitely many actions. Unfortunately, our functions, \emph{i.e.} the gain from trade and revenue, are \emph{not} continuous, making standard grids like the uniform one non-effective. Indeed, \citet{cesa2024bilateral} show that such grids provide a good approximation only under specific assumptions, \emph{e.g.}, bounded density, even in the simpler bilateral trade setting.

In the following, we provide two different grids depending on the functions that we are trying to approximate.
The first is an new adaptive grid useful to work with gain from trade.
The second is a non-uniform grid that takes inspiration from~\citet{bernasconi2024no}.

\subsection{An Adaptive Grid for $\gft$}





The $\gft(\cdot,\cdot)$ function is \emph{not} continuous with respect to the prices.
This lack of continuity makes impossible to design a \emph{finite} grid $\mathcal{B}$ of prices that approximate the continuous set of price without any knowledge of the underlying distribution $\mathcal{P}$.
Indeed, a priori it is impossible to know in which points such discontinuity lies.
We circumvent such problem introducing an adaptive grid built exploiting samples from the distribution $\mathcal{P}$.
Our main idea is to build a finer grid on the regions of the decision space that exhibit higher probability mass. These regions are where the most significant changes in the function occur.
Conversely, we use a coarser grid in regions where the probability density is lower, minimizing unnecessary exploration in less relevant areas.
Intuitively, we want to include in the grid the points to which distribution $\mathcal{P}$ assigns high probability, while guaranteeing that valuations in between two points of the grid occurs with small probability.
Such property mimics the bounded density assumption that in previous works has been use essentially to bound the probability that a valuation lies between two point in the uniform grid~\cite{cesa2023repeated}.
Finally, we notice that our approach is heavily based on our stronger asymmetric feedback and cannot be applied to bilateral trade with two-bit feedback.
Interestingly, we show that it is sufficient to use an adaptive grid only on the buyers' price, on which we have essentially full feedback committing to $q_t=0$.

As a technical tool, we use the DKW inequality~\citep{dvoretzky1956asymptotic}.
Given some samples from a distribution, this inequality provides high-probability bounds on the difference between the true probability of any interval and its empirically estimated probability.
By leveraging these bounds, we can ensure that the probability that the valuation lies between two points of our grid is small.
\Cref{alg: grid estimation} provides a pseudo-code of our approach to build the grid. The algorithm takes $T_0$ samples of the highest valuation $\bar b_t$. This can be achieved posting the prices $q_t=0$ and $p_t=\underline{b}_t$. Let $\mathcal{B}$ be the multi-set of observed valuations.
Then, the algorithm partitions $\mathcal{B}$ finding a set of points $p_1<p_2<\ldots<p_K$ such that
\[ |\left\{b \in \mathcal{B}: p_{i-1}<b<p_{i}\right\}|\le \frac{T_0}{K} \quad \forall i \in [K].\]
Finally, the algorithm selects a grid depending on the budget balance constraint.
In the strong budget balance setting, it is sufficient to discretize the buyers' prices, and hence we build a grid $\mathcal{B}^B$ that it the union of the computed grid and the uniform grid. The uniform grid is needed to guarantees that prices are not too far apart.
In the global budget balance setting, it builds the same grid for the buyers' prices and the uniform one for the seller's ones.

\begin{algorithm}[]\caption{\texttt{GridEstimation}}\label{alg: grid estimation}
\begin{algorithmic}[1]
		\State Input: $\textnormal{type}\in\{\SBB,\GBB\},T_0,K,\delta$
		\State Initialize $\mathcal{B}^B\gets\{0\}$, $k\gets0$, $p_0\gets0$ \Comment{B is  a multiset, \emph{i.e.}, it contains duplicates}
		\For{$t\in[T_0]$}
        \State Play $q_t=0$ and observe $\bar{b}_t,\underline{b}_t$
        \State Play $p_t=\underline{b}_t$
	  \State $\mathcal{B}\gets \mathcal{B} \cup \{\bar{b}_t\}$
		\EndFor
		\While{$k< K$}                             
		 \State $p_{k+1}\gets$ \textnormal{largest} $q\in \mathcal{B}$ such that $|\left\{b \in \mathcal{B}: p_k<b<q \right\}|\le T_0/K$ \Comment{$p_{k+1}\gets 1$ if such $q$ doesn't exist}
		\State $\mathcal{B}^B\gets \mathcal{B}^B\cup \{p_{k+1}\}$
		\State $k\gets k+1$ 
		\EndWhile
		\If{\textnormal{type} $=$ \SBB}
		 \State$\mathcal{B} \gets \mathcal{B}^B$
		\ElsIf{\textnormal{type} $=$ \GBB}
		\State $\mathcal{B}^S\gets \left\{p \in [0,1]: ~ p=\frac{n}{K}, n\in \{0\}\cup[K]\right\}$
		\State $\mathcal{B}^B\gets \mathcal{B}^B \cup \left\{q \in [0,1]: ~ q=\frac{n}{K}, n\in \{0\}\cup[K]\right\}$
		\State $\mathcal{B} \gets \mathcal{B}^S \times \mathcal{B}^B$
		\EndIf
		\State Return $\mathcal{B}$
	\end{algorithmic}
\end{algorithm}

Our first result shows that the gain from trade of the optimal price on the grid is closed to the optimal gain from trade of prices $q \in [0,1]$. Formally:
\begin{restatable}{lemma}{gridsbb}
\label{lemma: gridsbb}
	 Fix $T_0,K,\delta$ and let $\mathcal{B}= $\texttt{GridEstimation}$(\SBB,T_0,K,\delta)$. 
	  Then with probability at least $1-\delta$, it holds
	\begin{equation*}
		\max_{q \in [0,1]} \gft(q)-\max_{q\in \mathcal{B}}\gft(q) \le \frac{1}{K}+ \sqrt{\frac{\ln(\frac{2}{\delta})}{2T_0}} .
	\end{equation*}

\end{restatable}


We can prove a similar result for the global budget balance setting and couples of prices $(p,q)$.
Let $(p^*,q^*)\in [0,1]$ be an expected gain from trade maximizing couple of prices which satisfies $\GBB$, \emph{i.e.}, a solution to Program~\ref{prog:baseline}. Moreover, we let $(p_{k^*},q_{j^*})\in \mathcal{B}$ be the couple of prices where $p_{k^*}$ is the smallest element in $\mathcal{B}^S$ such that $p^*\le p_{k^*}$ and  $q_{j^*}$ is the smallest element in $\mathcal{B}^B$ such that $q^*\le q_{j^*}$.
Finally, we let $\zeta\coloneqq - \rev(p_{k^*},q_{j^*})$. We will see that $\zeta$ is small and its value depends on the assumption on the underlying valuations distribution.

Our next results show that $(p_{j^*},q_{k^*})$ achieves large revenue and it is almost budget balance, \emph{i.e.}, $\zeta$ is not too large.
%
Unfortunately, our result holds only for distributions that are independent between the seller and buyers or with bounded density. This is essentially unavoidable as we will observe in the lowerbound presented in \Cref{ex:impossible}.

\begin{restatable}{lemma}{gridgbb}
	\label{lemma: gridgbb}
	Fix $T_0,K,\delta$ and let $\mathcal{B}= $\texttt{GridEstimation}$(\GBB,T_0,K,\delta)$. 
	If $\mathcal{P}=\mathcal{P}^S \times \mathcal{P}^B $, then with probability at least $1-\delta$ it holds
	\begin{equation*}
		\gft(p^*,q^*)-\gft(p_{k^*},q_{j^*}) \le \frac{2}{K}+ \sqrt{\frac{\ln(\frac{2}{\delta})}{2 T_0}}.
	\end{equation*}
	Moreover, if $\mathcal{P}$ has density bounded by $M$, then with probability at least $1-\delta$ it holds
	\begin{equation*}
	\gft(p^*,q^*)-\gft(p_{k^*},q_{j^*})  \le \frac{M+1}{K}+ \sqrt{\frac{\ln(\frac{2}{\delta})}{2 T_0}}.
	\end{equation*} 
\end{restatable}

\begin{restatable}{lemma}{lemmaRevAlmostPositive}
	\label{lemma: -revk}
	Fix $T_0,K,\delta$ and let $\mathcal{B}=\texttt{GridEstimation}(GBB,T_0,K,\delta)$. If $\mathcal{P}=\mathcal{P}^S \times \mathcal{P}^B $, then with probability at least $1-\delta$ it holds
	\begin{equation*}
		\zeta\coloneqq -\rev(p_{k^*},q_{j^*})\le \frac{2}{K}+2\sqrt{\frac{\ln(\frac{2}{\delta})}{2T_0}},
	\end{equation*}
	Moreover, if $\mathcal{P}$ has density bounded by $M$, then with probability at least $1-\delta$ it holds
	\begin{equation*}
				\zeta\coloneqq - \rev(p_{k^*},q_{j^*})\le \frac{2M}{K}.
	\end{equation*}
\end{restatable}



\subsection{An Additive-Multiplicative Grid for Revenue}

When working with global budget balance algorithms, it is essential to have a way to generate revenue, which will be useful to explore and deal with our approximate estimations.
For instance, the process of exploring the decision space imposes a trade-off. In order to explore efficiently and gain useful information, the algorithm must sometimes sacrifice revenue (see \Cref{sec:GBBalgorithm}). Specifically, it may be required to make decision with immediate negative revenue in order to gather critical information that will help in the future decisions. This creates a tension between the need of accumulating revenue and the necessity of exploring. We deal with this problem accumulating enough revenue in the first rounds to absorb the loss of revenue incurred during the exploration phase.

In order to accumulate revenue, we propose the construction of a finite class of mechanisms, denoted as \( F_K \). This class is defined as the union of two distinct families of mechanisms: \( F_K^+ \) and \( F_K^- \), and it is designed to maximize the revenue.
%
%
%
Formally, we define the two classes of mechanisms:
\begin{itemize}
\item \emph{Mechanisms class $F_K^+$}: this class contains all the couples of prices $(p, p + 2^{-j})$, where $p \in \mathcal{B}^S$ and $j \in [\lceil \log_2(T) \rceil]$. These mechanisms set a reserve price  $q=p + 2^{-j}$ to the buyers, while offering to the seller price $p$, thereby retaining a revenue of at least $2^{-j}$ if the trade happens.
\item \emph{Mechanisms class $F_K^-$}: this class includes all mechanisms of the form $(\max\{q, \underline{b}_t\} - 2^{-j}, q)$, where $q \in \mathcal{B}^B$ and $j \in \left[\lceil \log_2(T) \rceil\right]$. Here, the reserve price is set as $q$, while the seller is offered a price derived from the auction outcome (the maximum of the reserve price and the second-highest bid) minus $2^{-j}$, ensuring a revenue of $2^{-j}$ if the trades happens.
\end{itemize}

Then, we let $F_K=F_K^+\cup F_K^-$.
We proceed showing that the optimal mechanism in $F_K$ extracts revenue that is a $O(\log(T))$-approximation of the optimal gain from trade. Intuitively, this shows that we are not incurring too much regret while accumulating revenue (playing the revenue maximizing mechanism in $F_K$).
In particular, we show that such result holds if the seller and buyer distributions are independent or the joint distribution $\mathcal{P}$ has bounded density.

\begin{restatable}{lemma}{revdisc} \label{lemma: rev opt disc}
    Let $\mathcal{B}= $\texttt{GridEstimation}$(GBB,T_0,K,\delta)$. Then,  if $\mathcal{P}=\mathcal{P}^S \times \mathcal{P}^B$ and/or $\mathcal{P}$ has density bounded by $M$, it holds
	\begin{equation*}
	 \gft(p_{k^*},q_{j^*}) \le 6 \log(T) \max_{(p,q)\in F_K}\rev(p,q) + 2/T +  \zeta .
	\end{equation*}
\end{restatable}

Now, we observe that the revenue maximizing prices on the grid can be learner easily using a general purpose regret minimizer since: 1) we have bandit feedback on the revenue, and 2) the size of the mechanism set satisfies $\lvert F_K \rvert \leq 3K \log(T)$. 


\begin{restatable}{lemma}{maxrevhedge}
	\label{lemma: maxrevhedge}
    Fix $T_0=\BigOL{T^{\nicefrac{2}{3}}},K=\BigOL{T^{\nicefrac{1}{3}}}$.
	There exists a learning algorithm \texttt{maxREV} that with probability at least $1-2\delta$ guarantees
	 \begin{equation*}
	 	\tau \cdot \gft(p_{k^*},q_{j^*}) \le 6 \log(T) \sum_{t\in [\tau]} \rev_t(p_t,q_t) + \BigOL{T^{2/3}} + \tau \zeta \quad \forall \tau \in [T]. 
	 \end{equation*}
     \end{restatable}


\section{Achieving No-Regret With Respect to \SBB~Mechanisms} \label{sec:SBB}

In this section we design an algorithm that, using weakly budget balance mechanisms, achieves regret of order $\BigOL{T^\frac{2}{3}}$ with respect to the optimal strongly (or weakly) budget balanced mechanism.

Such result is optimal in two aspect. First, while from a gain from trade perspective weakly budget balance mechanisms are dominated by strong budget balanced ones, they are necessary in learning settings. Indeed, \citet{cesa2024bilateral} shows that it is impossible to learn using the limited feedback of \SBB mechanisms  even in bilateral trade setting. Moreover, our regret rate is optimal even in the simpler bilateral trade setting~\citep{cesa2024bilateral}. 
Such lowerbound easily extends to our setting with multiple buyers. As we show in the following remark, it is sufficient to observe that in both our model and bilateral trade the feedback on the seller valuation is the same and that in the instances designed by \citep{cesa2024bilateral} the challenging task is to estimate the seller valuation distribution.

\begin{remark}
Even if \citet{cesa2024bilateral} provide a the regret lower bound $\Omega(T^{2/3})$ for strong budget balance mechanisms with $1$ buyers, such lower bound can be trivially extended to our setting. Indeed, the instances used by their proof differs only on the sellers distributions, where our setting present the same one-bit feedback, \emph{i.e.}, we only observe whether the seller accept or not the price. Hence, our asymmetric feedback does not provide any advantage in that specific instances.
Moreover, the price proposed to the seller $\max\{q,\underline{b}\}$ collapses to the one in \citet{cesa2024bilateral} whenever there is a single buyer, or equivalently all the buyers' valuations except one are $0$.
\end{remark}

Now, we combine the results in \Cref{sec:decomposition} and \Cref{sec:grid} to design our algorithm.
The algorithm works in three phases. In the first phase, it employs \Cref{alg: grid estimation} to build an adaptive grid. Then, it exploits 
\Cref{lemma: gft decomposition} to split the gain from trade estimation into two components. In the second phase, it learns ${\gft}_1(\cdot)$ for all the possible prices using uniform exploration. Finally, in the last phase, it exploits the bandit feedback on ${\gft}_2(\cdot)$ to run an UCB-like algorithm based on the estimates of ${\gft}_1(\cdot)$.
A pseudo-code of our approach is presented in \Cref{alg: SBB new real}.
More in details, in the first $T_0$ rounds the algorithm invokes Algorithm~\ref{alg: grid estimation} and receives the computed set $\mathcal{B}$ of prices, which contains at most $K$ elements.
During the next $T_0$ episodes, from episode $T_0+1$ to $2T_0$, the algorithm estimates the function $\gft_1(\cdot)$. This can be done sufficiently fast since it is possible to observe an unbiased estimator of all $\gft_1(q)$ at the same time.  
In this phase, the algorithm samples $U_t$ from a uniform distribution over $[0,1]$, and plays the mechanism $(U_t, U_t)$, where the reserve price is set to $U_t$ and the seller is offered the price $U_t$ as well.
This is the only phase in which the algorithm uses weakly budget balance mechanisms. Indeed, if $\underline{b}_t > U_t$ and $s_t\le U_t$ the mechanism designer has positive revenue.
At the end of this phase, the algorithm constructs an estimate of $\gft_1(\cdot)$ for each element of $q_k \in \mathcal{B}$:

\begin{equation*}
	\widehat{\gft}_1(q_k) = \frac{1}{T_0} \sum_{t=T_0+1}^{2T_0} \mathbb{I}(s_t \le U_t \le \max\{q_k, \underline{b}_t\}) \mathbb{I}(\overline{b}_t \ge q_k).
\end{equation*}

Notice that while the two indicators cannot be observed simultaneous, it is possible to determine if both indicator are true at the same time.
Then, Lemma~\ref{lemma: gft decomposition} guarantees that $\mathbb{E}[\widehat{\gft}_1(q_k)] = \gft_1(q_k)$ for all $q_k \in \mathcal{B}$.
Thus, applying an Hoeffding inequality, the algorithm constructs an high-probability upper confidence bound for $\gft_1$:

\begin{equation*}
	\overline{\gft}_1(p_k) = \widehat{\gft}_1(p_k) + \sqrt{\frac{\ln\left(\frac{2TK}{\delta}\right)}{2T_0}} \quad \forall q_k \in \mathcal{B},
\end{equation*}
which holds with probability $1-\delta$.

From episode $2T_0+1$ until the final episode $T$, the only remaining task is to estimate the second component of $\gft$, \emph{i.e.}, $\gft_2$.
It is not possible to use an explore than commit approach for this component, since estimating $\gft_{2}$ for different prices requires to choose different prices. Hence, it is impossible to compute all estimates at the same time (as for $\gft_1$), and such approach would incur suboptimal regret rates.
However, we can use a different property of $\gft_2$. Indeed, the learner observes bandit feedback on this component.
Hence, the algorithm can use an UCB-like technique for explore and exploit at the same time.
For each episode $t$, the algorithm builds an estimator:

\begin{equation*}
	\widehat{\gft}_{2,t}(q_k) = \frac{1}{n_t(q_k)} \sum_{i=2T_0+1}^t \left(  \overline b_i - \max\{q_k, \underline{b}_i\} \right) \mathbb{I}(\overline b_i \ge q_k, s_i \le \max\{q_k,\underline{b}_i\})\mathbb{I}(q_i=q_k) \quad \forall q_k \in \mathcal{B},
\end{equation*}
where $n_t(q)$ represents the number of episodes in which $q$ is played from round 1 to $t$.
Using an Azuma-Hoeffding’s inequality, the algorithm updates the optimistic upper confidence estimate for $\gft_2$:

\begin{equation*}
	\overline{\gft}_{2,t}(q_k) = \widehat{\gft}_{2,t}(q_k) + \sqrt{\frac{\ln\left(\frac{2TK}{\delta}\right)}{2(t - (K + T_0))}} \quad \forall q_k \in \mathcal{B}.
\end{equation*}

Finally, the algorithm selects the reserve price $q_{t+1}$ for the next episode as the one that maximizes the optimistic estimate $\overline \gft_1(q_k)+\overline \gft_2(q_k)$ (Line~\ref{line: alg1 line 15}). Combining all the components we obtain the following:

\begin{algorithm}[H]\caption{Regret Minimizer for \SBB~Mechanisms}\label{alg: SBB new real}
	\begin{algorithmic}[1]
		\State Input: $T_0$, $K$, $\delta$
		 \State$\mathcal{B} \gets \texttt{GridEstimation}(\SBB,T_0,K,\delta)$ \label{line: alg1 line 2}
		\State$\widehat{\gft}_1(p_k) \gets 0 \ \ \forall p_k \in \mathcal{B}$
		\For{$t=T_0+1,\ldots,2T_0$}
		 \State Sample $U_t \sim \mathcal{U}([0,1])$ \label{line: alg1 line 5} \Comment{sample $U_t$ from the uniform distribution over $[0,1]$}
		 \State Play $(p=U_t,q=U_t)$ \label{line: alg1 line 6}
		\For{$q_k \in \mathcal{B}$}
		 \State$\widehat{\gft}_1(q_k) \gets \widehat{\gft}_1(q_k) + \frac{1}{T_0}\mathbb{I}(s_t\le U_t \le \max\{q_k,\underline{b}_t\})\mathbb{I}(\bar{b}_t \ge q_k)$   \label{line: alg1 line 8}
	\EndFor
    \EndFor
		\State$\overline{\gft}_1(q_k)\gets \widehat{\gft}_1(q_k)  + \sqrt{\frac{\ln(\frac{2TK}{\delta})}{2T_0}}\quad \forall q_k \in \mathcal{B}$ \label{line: alg1 line 9}
		\State$\widehat{\gft}_2(q_k) \gets 0, \overline{\gft}_2(q_k)\gets1$ \label{line: alg1 line 10}
		\For{$t=2T_0+1,\ldots,T$}
		\State play $q_t$  \label{line: alg1 line 12} 
		\State play $p_t=\max\{ q_t, \underline{b}_t \}$  \Comment{price is posted only if $\bar b_t\ge q_t$}
        \State $\widehat{\gft}_{2,t}(q_t)\gets  \frac{1}{n_t(q_t)}\sum_{i=2T_0+1}^t(b_t-p_t)\mathbb{I}(\bar b_t \ge q_t, s_t\le p_t)$  \label{line: alg1 line 13} 
		\State $\overline{\gft}_{2,t}(p_t)\gets\widehat{GFT}_{2,t}(p_t) + \sqrt{\frac{\ln(\frac{2TK}{\delta})}{2(t-(K+T_0))}}$  \label{line: alg1 line 14}
	    \State$q_{t+1}\gets \argmax_{q_k\in \mathcal{B}}\overline{\gft}_1(q_k)+\overline{\gft}_{2,t}(q_k)$ \label{line: alg1 line 15}
        \EndFor	
    \end{algorithmic}
\end{algorithm}

\begin{restatable}{theorem}{sbbtheo}
    For any $\delta>0$, \Cref{alg: SBB new real} with $T_0= \BigOL{T^{\frac{2}{3}}}$ and $K=\BigOL{T^{\frac{1}{3}}}$ guarantees:
	\begin{equation*}
		R_T = \BigOL{T^{\frac{2}{3}}},
	\end{equation*}
    with probability at least $1-3\delta$.
\end{restatable}

\section{Achieving No-Regret With Respect to \GBB Mechanisms}

In this section, we focus on the more challenging global budget balance setting.
Our results hold when the buyers' and seller's distribution are independent or with bounded density, as required by \Cref{lemma: gridgbb} and \Cref{lemma: -revk}. As we show in \Cref{sec:open}, one of these properties is essential for a broad class of approaches.

Here, differently from the previous section we have to deal with unknown constraints. In particular, we have to estimate which mechanisms are indeed feasible.
Moreover, we have a much larger space of discretized mechanisms. Indeed, $|\mathcal{B}|=K^2$, while in the previous section we had to optimize over only $K$ \SBB mechanisms.
While this does not hurts the estimation of $\gft_1$, it makes impossible to obtained optimal learning rate only exporting the bandit feedback on $\gft_2$. 
Indeed, running a regret minimizer on $K^2$ arms will provide a $\tilde O(T^{3/4})$ regret bound.
Hence, the main challenge is to use the slightly more informative feedback on $\gft_2$. In particular, such problem can be seen as an MAB with feedback graph, which independent  is $K$. However, we remark that our setting presents constraints which make standard algorithm for MABs with feedback graph nonapplicable.

This section is organized as follows.
In \Cref{sec:feedback}, we provide an algorithm for constrained MABs with feedback graph. While it uses some specific property of our problem, we believe that it could be of general interest and can be extended to other settings.
%
In \Cref{sec:GBBalgorithm}, we exploit such regret minimizer to design our \GBB regret minimizer algorithm.

\subsection{Constrained MABs With Feedback Graphs}  \label{sec:feedback}

We focus on a generalization of MABs with feedback graph to constrained settings. 
At each round $t \in [t]$, the learner selects an arm $a_t$ from a finite set $\mathcal{A}$.
Then, random rewards $r_t(a)\in [-1,1]$ and costs $c_t(a)\in [-1,1]$ for each arm $a$ are sampled from a distribution $\mathcal{D}$. We let $r(a)$ and $c(a)$ be the expected reward and cost of arm $a$.
Different from standard MABs, the learner observe the feedback of some arms $\textnormal{Obs}_t$ which depend on the pulled arm $a_t$. We will introduce later the specific feedback we will focus on.

Let $a^*$ be an optimal solution of
\begin{align*}
    \max_{a\in \mathcal{A}} r(a) \quad \textnormal{s.t.} \quad c(a)\le 0.    
\end{align*}

We aim at designing algorithms that guarantees sublinear regret and violation of the constraints. Differently from standard MABs, we aim at obtaining a better dependence than $\sqrt{|A|}$ with respect to the number of arms, exploiting the enhance feedback.

Some general approaches for addressing our problem can be derived from the literature on online learning with long-term constraints.
In particular, some works derive general primal-dual methods that uses a black-box regret minimizer for the \emph{uncontrained} problem~\cite{castiglioni2022unifying,slivkins2023contextual,bernasconibandits}. In our specific setting, we could use any regret minimizer for MABs with feedback graphs (see, \emph{e.g.}, \citep{lykouris2020feedback}).
Unfortunately, such approach provides optimal guarantees only assuming Slater condition. Such condition is not verified in our setting, since there does not exists a couple of prices that guarantees (constant) strictly-positive revenue.
In the following, we design a regret minimizer for a more restricted setting but which does not require Slater condition.
Moreover, differently from the mentioned primal-dual method, our algorithm bounds the \emph{positive} violations, which can be of independent interest.

\paragraph{Multi-line Feedback Graphs with monotone Costs}
We consider a particular feedback graph that generalizes the one of our problem. The set of arms $\mathcal{A}$ includes an arm $a_{i,j}$ with $i\in [n]$ and $j\in [m]$, where $n,m \in \mathbb{N}$.
When we pull arm $a_{i,j}$, we observe the reward and costs of all arms $a_{i,k}$ for all $k\ge j$. We call a set of arms $ \{a_{i,j}\}_{j \in [m]}$ a \emph{line}.
Notice that such graph has independent number $n$ (see \citep{alon2015online}  for a definition). Hence, in standard unconstrained problem it is possible to obtain regret $\sqrt{T n}$ with standard techniques.
Of course, we such results does not extend to constrained settings.

In the following, we show how to extend the result to constrained setting when the following \emph{monotonicity} assumption on the cost holds:
\[c(a_{i,j})\ge c(a_{i,k}) \quad \forall i,j<k. \]
As we will see in the following, this assumption is satisfied by our setting. 

\Cref{alg:successive} provides a pseudo-code of our algorithm. It essentially implements a successive arm elimination algorithm with \emph{precedences}.
It keeps track of the arms that are feasible and not dominated into a set $\mathcal{S}_t$. Then, it chooses the line with the better upper bound on the reward and it plays the arm in such line (and in $\mathcal{S}_t$) which provides stronger feedback. In particular, at each round it chooses the arm $a_{i_t,\bar j_t}\in \mathcal{S}_t$ that maximizes the upper bound on the reward. Then, chooses the arm $a_{i_t,j_t}\in \mathcal{S}_t$ which is not provably suboptimal.
This is done taking track of a lower and upper confidence bound on each arm as follows:

\begin{subequations}\label{eq:UCBlike}
\begin{align}\
    &\bar r_t(a)= \frac{1}{n_t(a)}\sum_{i=1}^t r_i(a)\mathbb{I}(a\in \textnormal{Obs}_i) + 2\sqrt{\frac{\ln(\frac{2nmT}{\delta})}{2n_t(a)}} \quad \forall a\in \mathcal{A}\\
    &\underline{r}_t(a)= \frac{1}{n_t(a)}\sum_{i=1}^t r_i(a)\mathbb{I}(a\in\textnormal{Obs}_i) - 2\sqrt{\frac{\ln(\frac{2nmT}{\delta})}{2n_t(a)}} \quad \forall a\in \mathcal{A}\\
    & \underline{c}_t(a)= \frac{1}{n_t(a)}\sum_{i=1}^t c_i(a)\mathbb{I}(a\in \textnormal{Obs}_i) - 2\sqrt{\frac{\ln(\frac{2nmT}{\delta})}{2n_t(a)}} \quad   \forall a\in \mathcal{A}
\end{align}
\end{subequations}
where $\textnormal{Obs}_t$ is the set of arms whose feedback is observers at episode $t$ and $n_t(a)= \sum_{i=1}^t\mathbb{I}(a\in \textnormal{Obs}_i)$.
At the end of each round, the algorithm removes from $\mathcal{S}_t$ all arms that are dominated or unfeasible.

The main property of our algorithm is that we always observe the reward of all the non-removed arms in the played line.
This can be done safely thanks to our monotonicity assumption. Indeed, monotone costs plays a crucial role in avoiding the removal of an optimal arm. Intuitively, if we remove an arm it means that there is an arm with larger index with higher utility. If it turns out that such arm is unfeasible, by the monotonicity of the costs it must be the case that also the first arm was unfeasible.

Finally, we highlight that our algorithm provides a stronger guarantees on the constraints then the one actually required. In particular, it guarantees that  $\sum_{t=1}^\tau [c(a_t)]^{+}$  is sublinear, where we define $[x]^{+}= \max\{x,0\}$. Formally:

\begin{restatable}{theorem}{theoremGMAB}
\label{theoremGMAB}
Algorithm~\ref{alg:successive} is such that with probability at least $1-2\delta$
    \begin{equation*}
        \tau \cdot r(a^*) - \sum_{t=1}^\tau r(a_t)  \le \BigOL{ \sqrt{nT}\log(m/\delta)} \forall \tau \in [T]
    \end{equation*}
    and 
    \begin{equation*}
        \sum_{t=1}^\tau [c(a_t)]^{+} \le \BigOL{\sqrt{nT} \log(m/\delta)} \forall \tau \in [T].
    \end{equation*}
\end{restatable}

\begin{algorithm}[H]\caption{Successive Arm Elimination with Precedences (\texttt{SAE-P})}\label{alg:successive}
\begin{algorithmic}[1]
\State Input: Set of arm $\mathcal{A}$, probability $\delta$
\State $\mathcal{S}_0(i)\gets \{a_{i',j}\in \mathcal{A}: i'=i\} $ for each $i\in [n]$
\State$\mathcal{S}_0\gets \mathcal{A}$
\State $\bar r_0(a),\underline{r}_0(a), \underline{c}_0(a) \gets 0$ for each $a\in \mathcal{A}$
\For{$t\in [T]$}
\State $a_{i_t,\bar j_t} \gets \argmax_{a \in \mathcal{S}_{t-1}} \bar r_{t-1}(a)$
\State $ j_t \gets \textnormal{minimum} \ j \ \textnormal{s.t.}\  a_{i_t,j}\in \mathcal{S}_{t-1} \land \left(\bar r_{t-1}(a_{i,j}) \ge \underline r_{t-1}(a) \ \forall a \in \mathcal{S}_{t-1} \right)  $
\State Play $a_{i_t, j_t}$ and observe feedback $(r_t(a_{i_t,j})$ for each $j \ge  j_t$
\State Compute $\bar r_t$ and $\bar c_t$ for all elements in $\mathcal{S}_{t-1}$ according to \Cref{eq:UCBlike}
\State $S_t(i_t) \gets \{a_{i_t,j}: j\ge j_t \land \underline{c}_t(a_{i_t,j})\le 0$
\State $S_t(i) \gets S_{t-1}(i)$ for all $i \neq i_t$ 
\State $S_t \gets \cup_{i\in [n]}S_t(i)$
\EndFor
\end{algorithmic}
\end{algorithm}

\subsection{A \GBB Regret-Minimizing Algorithm} \label{sec:GBBalgorithm}

In this section, we present a no-regret algorithm for learning optimal \GBB mechanisms. As previously noted, identifying which pairs of prices satisfy the budget balance constraint is not straightforward. While we partially address this challenge using Constrained Multi-Armed Bandits (MABs), this approach still incurs a sublinear constraint violation.

To compensate for the potential negative revenue in the initial rounds, we begin by accumulating revenue in the first rounds. Specifically, Algorithm \ref{alg: GBB not ind} starts with an initial phase dedicated solely to maximizing revenue. This allows the algorithm to reach a predefined budget threshold $\beta$, an input parameter, before transitioning to the exploration phase.

\texttt{RevMax} Algorithm is a generic state-of-the-art algorithm that maximize the Revenue over the the grid $F_K$, \emph{e.g.}, UCB1.
As shown in \Cref{lemma: maxrevhedge}, maximizing revenue over $F_K$ can be formulated as a stochastic MAB.
Setting $T_0=T^{2/3}$ and $K=T^{1/3}$, \texttt{RevMax} accumulates budget $\beta$, while suffering  regret at most $\BigOL{\beta + T^\frac{2}{3}}$ thanks to Lemma~\ref{lemma: rev opt disc}.
In the second phase, we estimate the first component of the revenue, 
$\gft_1$, as defined in \Cref{lemma: gft decomposition}.

In the final phase, we apply \Cref{alg:successive}. We give as input to \texttt{SAE-P} the cost
$c_t(p,q)$ and reward $r_t(p,q)$  for each arm for where feedback is available.
To ensure the feasibility of $p_{j^*},q_{j^*}$ within the Constrained MAB framework, we shift the revenue, \emph{i.e.}, minus the cost, by an upper bound $\bar \zeta$ which depend from the underlying distributional assumption (see \Cref{lemma: -revk})).
Finally, we set the reward $r_t$ to the sum of the upper bound on the already estimated component $\overline{\gft}_1$ and $\gft_{2,t}$ at the current round.

\begin{algorithm}[H]\caption{Regret Minimizer for \GBB Mechanisms}\label{alg: GBB not ind}
\begin{algorithmic}[1]
		\State Input: 
        $K,T_0,\delta,\beta,\bar\zeta$
		\State $\mathcal{B} \gets \texttt{GridEstimation}(GBB,T_0,K,\delta)$ \label{line: alg gbb line 2} 
		\State Run $\texttt{maxREV}(F_K,\beta)$
	\If{\texttt{maxREV} terminates at time $\tau<T-T_0$}
			\State $\widehat{GFT}_1(p_k,q_j) \gets 0$ for all                  $(p_k,q_j) \in \mathcal{B}$
			\For{$t=\tau+1,\ldots,\tau+T_0$}
				\State Sample $U_t \sim \mathcal{U}([0,1])$
				\State Play $(p_t=U_t,q_t=0)$
				\For{$p_k,q_j \in \mathcal{B}$}
					 \State $\widehat{GFT}_1(p_k,q_j) \gets \widehat{GFT}_1(p_k,q_j) + \frac{1}{T_0}\mathbb{I}(s_t\le U_t \le p_k)\mathbb{I}(\bar{b}_t \ge q_j)$  
				\EndFor
			\EndFor
            \State $\overline{\gft}_1(p,q) \gets \widehat{\gft_1}(p,q) + 2\sqrt{\frac{\ln(\frac{|\mathcal{B}|}{\delta})}{2T_0}}\quad \forall(p,q)\in \mathcal{B}$
            \State Instantiate an instance of Algorithm \texttt{SAE-P}$(\mathcal{A}=\mathcal{B},\delta)$
            \For{$t=\tau+T_0+1,\ldots,T$}
            \State Play prices $(p_t,q_t)$ proposed by \texttt{SAE-P}
            \State Observe feedback and build for all $q\in \mathcal{B}^B$ such that $q\ge q_t$:
            \[\gft_{2,t}(p_t,q)\gets (\bar{b}_t-p_t)\mathbb{I}(\bar{b}_t\ge q)\mathbb{I}(s_t\le p_t)\]  \[\rev_t(p_t,q) \gets (\max\{q,\underline{b}_t\}-p_t)\mathbb{I}(\bar{b}_t\ge q)\mathbb{I}(s_t\le p_t)\] 
            \State Build for all $(p,q)\in \mathcal{B}$ such that $q\ge q_t$ and $p=p_t$ 
            \[c_t(p,q)\gets \frac{1}{1+\bar\zeta}\left(-\rev_t(p,q) - \bar\zeta \right)\]\[r_t(p,q)\gets\frac{1}{2} \left(\overline{\gft}_1(p,q)+ \gft_{2,t}(p,q)\right)\]
            \State Update \texttt{SAE-P} with $c_t,r_t$
            \EndFor
	\EndIf
    \end{algorithmic}
\end{algorithm}

Combining all these components, we can prove that \Cref{alg: GBB not ind} guarantees regret of order $\BigOL{T^{\frac{2}{3}}}$ and satisfies the \GBB contraint. Formally:

\begin{restatable}{theorem}{gbbtheo}
    Assume that $\mathcal{P}= \mathcal{P}^S\times \mathcal{P}^B$ or that $\mathcal{P}$ has bounded density.
    Then, for any $\delta>0$, \Cref{alg: GBB not ind} with $K=\BigOL{T^\frac{1}{3}}$, $T_0=\BigOL{T^\frac{2}{3}}$, $\beta= \BigOL{T^{2/3}}$, and $\bar\zeta =\BigOL{T^{-1/3}}$ guarantees:
	\begin{equation*}
		R_T = \BigOL{T^\frac{2}{3}} \quad \textnormal{and} \quad \sum_{t=1}^T\rev_t(p_t,q_t) \ge 0,
	\end{equation*} 
    with probability at least $1-6\delta$.
\end{restatable}

We remark that in the theorem we highlight only the dependence of $T$. For instance, we omitted the dependence from $M$ for bounded density distributions. See \Cref{app:GBB} for more details.

\section{Open Problem: Assumption-Free  Regret Minimization With \GBB Constraints} \label{sec:open}

In the previous section, we design a \GBB algorithm under some assumptions on the underling valuation distribution.
However, it remains an open question whether a sublinear regret bound can be achieve without these assumptions. 

We conclude that paper providing some evidence that our approach is doomed to fail, and designing such algorithm requires novel non-trivial techniques. Specifically, the following example shows that it is impossible to learn an approximately optimal couple of \emph{fixed} prices $(p,q)$.
Since this is the aim of our algorithm (see \Cref{sec:GBBalgorithm}), it is clear that our algorithm cannot be easily extended to settings without assumptions.
Notice that this does not rule out the existence of an algorithm achieving sublinear regret. However, such an algorithms would be much more complex as it should dynamically adjust $q_t$ based on the buyers valuations $\bvec_t$. 

\begin{example}\label{ex:impossible}
We construct a family of distributions where learning the optimal price pair $(p^*,q^*)$ is as hard as finding an arbitrary small segment on a line. 
In particular, consider a family of seller/buyers distributions $\mathcal{P}_\epsilon^x$ parametrized by $x \in (7/16,9/16)$ and $\epsilon\in(0,1/16)$ described as follows: 
\begin{align*}
\mathbb{P}((s_t,\bar{b}_t,\underline{b}_t)=(a,y,y'))=
\begin{cases}
1/3 & \text{if } (a,y,y') = (x-\epsilon,3/4,0) \\
1/3 & \text{if }  (a,y,y') = (x+\epsilon,1/4,0)\\
1/3 & \text{if }  (a,y,y') = (0,1/4,0)\\
0 & \text{otherwise}
\end{cases}
\end{align*}
For valuations $(x-\epsilon,3/4,0)$ the trade happen if $ (p,q)\in [x-\epsilon,1]\times[0,3/4]$ with $\gft=3/4-x+\epsilon>3/16$, for  valuations $(x+\epsilon,1/4,0)$ the trade happens if $(p,q)\in [x+\epsilon,1]\times [0,1/4]$ with $\gft=1/4-x-\epsilon<-3/16$, and for valuations $(0,1/4,0)$ the trade happens if $(p,q)\in [0,1]\times [0,1/4]$ with $\gft=1/4$. Thus, every couple of prices $(p,q)\notin[x-\epsilon,x+\epsilon)\times [0,1/4]$  leads to an expected sub optimality gap of at least $1/16$. Hence, finding an approximately optimal mechanisms requires identify the arbitrarily small interval $[x-\epsilon,x+\epsilon)$ with one-bit feedback, \emph{i.e.}, if the seller accept the trade. This is equivalent to a binary search problem, which requires $\Omega\left(\log(\frac{1}{\epsilon})\right)$ samples.
Such problem requires an arbitrary large number of samples if $\epsilon$ is arbitrary small.
\end{example}

We leave open the question of whether this problem remains  hard for the broader class of mechanisms available to our regret minimizer or if more sophisticated mechanisms can circumvent this impossibility result.

\section*{Acknowledgments}
This paper is supported by the Italian MIUR PRIN 2022 Project ``Targeted Learning Dynamics:
Computing Efficient and Fair Equilibria through No-Regret Algorithms'', by the FAIR (Future
Artificial Intelligence Research) project, funded by the NextGenerationEU program within the PNRRPE-AI scheme (M4C2, Investment 1.3, Line on Artificial Intelligence), and by the EU Horizon project
ELIAS (European Lighthouse of AI for Sustainability, No. 101120237).

\bibliographystyle{ACM-Reference-Format}
\bibliography{References}

\clearpage
\appendix

\section{Proofs Omitted from \Cref{sec:decomposition}}

\gftdecompose*
\begin{proof}
	For all posted couple of prices $(p,q)\in [0,1]^2$ and realized valuations $(\bvec,s)$, the gain from trade can be rewritten as
	\begin{align*}
		(\bar{b}_t-s_t)\mathbb{I}(\bar{b}_t\ge q)\mathbb{I}(s_t\le p) & = (\bar{b}_t-p+p-s_t)\mathbb{I}(\bar{b}_t\ge q)\mathbb{I}(s_t\le p)\\
		& = (\bar{b}_t-p)\mathbb{I}(\bar{b}_t\ge q)\mathbb{I}(s_t\le p)+(p-s_t)\mathbb{I}(\bar{b}_t\ge q)\mathbb{I}(s_t\le p).
	\end{align*}
	Taking the expectation over $(\bvec,s)\sim \mathcal{P}$ we get
	\begin{equation*}
		\mathbb{E}[	\gft_t(p,q)] = \gft_1(p,q) + \gft_2(p,q).
	\end{equation*}
    This prove the first part of the statement.
    
    Now, observe that the probability that a random variable $U$ drawn  uniformly from $[0,1]$ is in an interval $[x,y]$, $x,y\in [0,1]$, is \[\mathbb{P}(x\le U\le y)=x-y.\]
    Therefore, recalling that we defined as $\mathcal{U}$ the uniform distribution over $[0,1]$, for each valuations $(\bvec,s)$ it holds
	\begin{align*}
		\mathbb{P}_{U\sim \mathcal{U} }(s\le U\le p,\bar{b}\ge q)& =  (p-s) \mathbb{I}(s\le p,\bar{b}\ge q). 
	\end{align*}
    Taking the expectation over $(\bvec,s)\sim \mathcal{P}$ and recalling the definition of $\gft_1(p,q)$ the statement easily follows.
\end{proof}

\gftdecomposeSBB*
\begin{proof}
This result can be proven by a completely analogous proceeding to the proof of Lemma~\ref{lemma: gft decomposition}.
\end{proof}

\section{Proofs Omitted from \Cref{sec:grid}} \label{app:grid}

Before proving the results on our discretization grids, we state this useful concentration bound from~\citet{cesa2024bilateral}

\begin{theorem}[Theorem 14 in \citep{cesa2024bilateral}]
\label{lemma: DKW}
    If $(\Omega,\mathcal{F},\mathbb{P})$ is a probability space and $(X_n)_{n\in \mathbb{N}}$ is a $\mathbb{P}$-i.i.d. sequence of random variables, then , for any $\epsilon>0$ and all $m\in \mathbb{N}$, it holds
    \begin{equation*}
		\mathbb{P}\left[\sup_{x\in \mathbb{R}}\bigg\lvert\frac{1}{m}\sum_{k=1}^m\mathbb{I}\{X_k\le x\}-\mathbb{P}[X_1\le x]\bigg\rvert > \epsilon\right] \le 2 \exp (-2 m \epsilon^2).
	\end{equation*}
\end{theorem}

\gridsbb*

\begin{proof}
Let $q^*\in [0,1]$ be the price that maximizes the gain from trade.
If $q^*$ belong to the grid $\mathcal{B}^B$ the statement trivially follows.
Otherwise, by the definition of $\mathcal{B}^B$, there exists a $q_k\in \mathcal{B}^B$ such that $\frac{1}{T_0}\sum_{t=1}^{T_0}\mathbb{I}(  q^*\le \bar{b}_t< q_k )\le \frac{1}{T_0 }\frac{T_0}{ K}=\frac{1}{K}$.
Then, applying \Cref{lemma: DKW} with random variable $X_t=\bar{b}_t$ 
we obtain:
\[\bigg| \mathcal{P}_{(\bvec,s)\sim \mathcal{P}}( \bar b\le q^*)- \frac{1}{T_0}\sum_{t=1}^{T_0}\mathbb{I}(\bar{b}_t\le q^* )  \bigg|\le \sqrt{\frac{\ln(\frac{2}{\delta})}{2 T_0}}   \]
and
\[\bigg| \mathcal{P}_{(\bvec,s)\sim \mathcal{P}}( \bar b\le q_k)- \frac{1}{T_0}\sum_{t=1}^{T_0}\mathbb{I}(\bar{b}_t\le q_k )  \bigg|\le \sqrt{\frac{\ln(\frac{2}{\delta})}{2 T_0}}       \]
with probability at least $1-\delta$.

Hence,
\begin{align} \label{eq:close}
\mathcal{P}_{(\bvec,s)\sim \mathcal{P}}(\bar{b}_t\le q_k)- \mathcal{P}_{(\bvec,s)\sim \mathcal{P}}(\bar{b}_t\le q^*) &\le \frac{1}{T_0}\sum_{t=1}^{T_0}\mathbb{I}(  q^*\le \bar{b}_t< q_k ) + 2 \sqrt{\frac{\ln(\frac{2}{\delta})}{2 T_0}}\nonumber\\ &\le \frac{1}{K}+ 2 \sqrt{\frac{\ln(\frac{2}{\delta})}{2 T_0}}.
\end{align}

Now, let's split the valuations $(\bvec,s)$ into three sets:
\begin{itemize}
    \item $V_1$ such that $\bar b \ge q_k$,
    \item $V_2$ such that $\bar b \in [q^*,q_k)$,
    \item $V_3$ such that $\bar b < q^*$.
\end{itemize}

It is easy to see that when the sampled valuations belongs to $V_1$ the gain from trade of $q_k$ is larger or equal to the one of $q^*$.
Moreover, when the sampled valuations belong to $V_3$ the gain from trade of both mechanisms is $0$.
Finally, when the sampled valuations belong to $V_2$ the mechanism $q^*$ might get a positive gain from trade (at most 1), while with $q_k$ the trade does not happen (with gain from trade $0$.
Hence, we can bound
\[ \gft(q^*)-\gft(q_k)\le P_{(\bvec,s)\sim \mathcal{P}}((\bvec,s)\in V_2)\le \frac{1}{K}+ 2 \sqrt{\frac{\ln(\frac{2}{\delta})}{2 T_0}},  \]
where the last inequality follows from \Cref{eq:close}

\end{proof}

\gridgbb*
\begin{proof}

    We bound the gain from trade in two steps. First, we show how a $(p_{k^*},q^*)$ on the grid well approximate the optimal couple $(p^*,q^*)$.
    The proof of this statement is split into two different proofs depending on the underlying assumption.
    Then, we provide a joint argument for the approximation of $(p_{k^*},q^*)$ with a $(p_{k^*},q_{j^*})$ on the grid. 
	
	\paragraph{ Independent distributions ($\mathcal{P}=\mathcal{P}^S \times \mathcal{P}^B$)} 
    We focus on the non-trivial case in which it does \emph{not} exists a couple of optimal prices $p^*,q^*$ such that $p^*\in \mathcal{B}^S$.
    Consider then the optimal prices such that the probability density function of the valuations of the sellers is strictly greater than zero in $p^*$\footnote{By convention we say that a non-absolute random variable $X$ that has non-zero probability in the set $\{x\} $ of null Lebesgue measure present a density function $f_X(x)=+\infty$}. Observe that such $p^*$ must exists if $p^*=0\in \mathcal{B}^S$ is \textbf{not} an optimal price. Then, to be optimal $(p^*,q^*)$ must be such that it exist an $\epsilon'\in (0,1)$ such that $\forall\epsilon\in (0,1)$ such that $\epsilon<\epsilon'$  \begin{equation*}
        \int_{p^*-\epsilon}^{p^*}\mathbb{E}[(\bar{b}-s)\mathbb{I}(\bar{b}\ge q^*)|s]f_s(s)ds\ge 0,
    \end{equation*} 
    since all prices $[0,p^*]\times \{q^*\}$ are feasible with expected revenue greater or equal to the expected revenue of $p^*,q^*$. This implies, as $\epsilon$ is arbitrarily small and the function $g:[0,1]\rightarrow [0,1]$ defined as $g(x)\coloneq \mathbb{E}[(\bar{b}-x)\mathbb{I}(\bar{b}_t\ge q^*)]$ is continuous, that 
    \begin{equation*}
        \mathbb{E}[(\bar{b}-s)\mathbb{I}(\bar{b}\ge q^*)|s=p^*]f_s(p^*)\ge0,
    \end{equation*}
    and finally, since $f_s(p^*)>0$ 

    \begin{align}\label{eq:posGFT}
    \mathbb{E}[(\bar{b}-p^*)\mathbb{I}(\bar{b}\ge q^*)] \ge 0.
    \end{align}

    Hence:

    \[\gft(p^*,q^*) =\mathbb{E}_{(\bvec,s)\sim \mathcal{P}} [(\bar{b}-s) \mathbb{I}(\bar{b}\ge q^*, s\le p^*)]\ge 0. \]

    Thus,

    \begin{align*}
        \mathbb{E}[\gft(p^*,q^*)&-\gft(p_{k^*},q^*)]\\
        &=  -\mathbb{E}[(\bar{b}-s)\mathbb{I}( p^*\le s < p_{k^*} )\mathbb{I}(\bar{b}\ge q^*)]  \\
        & \le  -\mathbb{E}[(\bar{b}-p_{k^*})\mathbb{I}(p^*\le s < p_{k^*})\mathbb{I}(\bar{b}\ge q^*)]\\
        & \le  -\mathbb{E}[(\bar{b}-p^*)\mathbb{I}(\bar{b}_t\ge q^*)] \mathbb{P}(p^*\le s < p_{k^*})  \\
        & - \mathbb{E}[(p^*-p_{k^*})\mathbb{I}(\bar{b}_t\ge q^*)] \mathbb{P}(p_{k^*} < s \le p^* )\\
        & \le 0 + \frac{1}{K}, 
    \end{align*}
    where the first inequality follows since the function is non-zero only for $p^* \le s < p_{k^*}$ and the last from \Cref{eq:posGFT} and since $-(p^*-p_{k^*})\le \frac{1}{K}$ by construction.

  \paragraph{$\mathcal{P}$ has density bounded by $M$} Similarly to the independent distribution case, we observe that the gain from trade of $(p^*,q^*)$ and $(p_{k^*},q^*)$ differs only if  $ p_{k^*} < s \le p^*$.
  Here, we exploit the bounded density assumption to derive
  \[ \mathbb{P}_{(\bvec,s)\sim \mathcal{P}}( p_{k^*} < s \le p^*)\le M/K. \]
  Hence:
	\begin{align*}
		\mathbb{E}[\gft_t(p^*,q^*)-\gft_t(p_{k^*},q^*)]   \le \mathbb{P}(p^*\le  s_t< p_{k^*}  ) \le \frac{M}{K}.
	\end{align*}

    Now, we show that $(p_{k^*},q_{j^*})$ has a gain from trade similar to $(p_{k^*},q^*)$ concluding the proof.
    Consider the non-trivial case $q^*\notin \mathcal{B}^B$.
    Then, similarly as we proved in \Cref{eq:close} of \Cref{lemma: gridsbb},
    we can prove that with probability at least $1-\delta$ it holds:
    \begin{align*} 
\mathcal{P}_{(\bvec,s)\sim \mathcal{P}}(\bar{b}_t\le q^*)- \mathcal{P}_{(\bvec,s)\sim \mathcal{P}}(\bar{b}_t\le q_{j^*}) \le \frac{1}{K}+ 2 \sqrt{\frac{\ln(\frac{2}{\delta})}{2 T_0}}.
\end{align*}
Moreover, similarly to the proof of \Cref{lemma: gridsbb}  we observe that the gain from trade of the two mechanisms differs only when $\bar b \in [q^*,q_{j^*})$, and such difference is at most $1$. 
Hence,

	\begin{align*}
		\gft(p_{k^*},q^*)-\gft_t(p_{k^*},q_{j^*}) \le P_{(\bvec,s)\sim \mathcal{P}}(\bar b \in [q^*,q_{j^*}))  \le \frac{1}{K} + \sqrt{\frac{\ln(\frac{2}{\delta})}{2 T_0}} ,
	\end{align*}
\end{proof}

Next, we show that has $(p_{k^*},q_{j^*})$ has a not too negative revenue, making it a viable option in the design of \GBB mechanisms.

\lemmaRevAlmostPositive*
\begin{proof}
    We divide the proof in two parts depending on the assumption on the distribution.

    \paragraph{ Independent distributions ($\mathcal{P}=\mathcal{P}^S \times \mathcal{P}^B$)}
    First, we notice that, since $(p^*,q^*)$ has expected non-negative revenue and sellers and buyers valuations are independent, then
    \begin{align*}
        0 &\ge \rev(p^*,q^*) = \mathbb{E}_{(s,\bvec)\sim \mathcal{P}}[(\max\{\underline{b},q^*\}-p^*)\mathbb{I}(s\le p^*)\mathbb{I}(\bar{b}\ge q^*)]\\
        & = \mathbb{E}_{(s,\bvec)\sim \mathcal{P}}[(\max\{\underline{b},q^*\}-p^*)|s\le p^*,\bar{b}\ge q^*]\mathbb{P}(s\le p^*)\mathbb{P}(\bar{b}\ge q^*)\\
        & = \mathbb{E}_{\bvec\sim \mathcal{P}^B}[(\max\{\underline{b},q^*\}-p^*)|\bar{b}\ge q^*]\mathbb{P}(s\le p^*)\mathbb{P}(\bar{b}\ge q^*)].
    \end{align*}

	which implies that either
	\begin{align}\label{eq:positiveRev}
		\mathbb{E}_{\bvec\sim \mathcal{P}^B}\left[\max\{\underline{b},q^*\}-p^*|\bar{b}_t\ge q^*\right]\ge 0,
	\end{align}
    or $\mathbb{P}(s_t\le p^*)\mathbb{P}(\bar{b}_t\ge q^*)=0$. Notice that the second case is trivial as it would imply $\mathbb{P}(s_t\le p_{k^*})\mathbb{P}(\bar{b}_t\ge q_{j^*})=0$ and therefore $\rev(p_{k^*},q_{j^*})=0$.

	Thus, considering the non-trivial case $\mathbb{P}(s_t\le p^*)\mathbb{P}(\bar{b}_t\ge q^*)>0$, with probability at least $1-\delta$ it holds
	\begin{align*} 
	&\rev(p_{k^*},q_{j^*})-\rev(p^*,q^*)\\
        & = \mathbb{E}[(\max\{q_{j^*},\underline{b}\}-p_{k^*})\mathbb{I}(\bar{b}\ge q_{j^*})\mathbb{I}(s\le p_{k^*})] -\mathbb{E}[(\max\{q^*,\underline{b}\}-p^*)\mathbb{I}(\bar{b}\ge q^*)\mathbb{I}(s\le p^*)]\\
        & \ge \mathbb{E}[(\max\{q^*,\underline{b}\}-p^*)\mathbb{I}(\bar{b}\ge q_{j^*})\mathbb{I}(s\le p_{k^*})] + (p^*-p_{k^*}) -\mathbb{E}[(\max\{q^*,\underline{b}\}-p^*)\mathbb{I}(\bar{b}\ge q^*)\mathbb{I}(s\le p^*)]\\
        &  \ge \mathbb{E}[(\max\{q^*,\underline{b}\}-p^*)\mathbb{I}(\bar{b}\ge q^*)\mathbb{I}(s\le p_{k^*})] - \mathbb{E}[(\max\{q^*,\underline{b}\}-p^*)\mathbb{I}(q^*\le \bar{b}< q_{j^*})\mathbb{I}(s\le p_{k^*})]\\
        &\hspace{6cm}+ (p^*-p_{k^*})-\mathbb{E}[(\max\{q^*,\underline{b}\}-p^*)\mathbb{I}(\bar{b}\ge q^*)\mathbb{I}(s\le p^*)]\\
        & \ge \mathbb{E}[(\max\{q^*,\underline{b}\}-p^*)\mathbb{I}(\bar{b}\ge q^*)]\mathbb{P}(s\le p_{k^*})-\mathbb{P}(q^*\le \bar{b}< q_{j^*}) \\
        &\hspace{6cm}+(p^*-p_{k^*}) -\mathbb{E}[(\max\{q^*,\underline{b}\}-p^*)\mathbb{I}(\bar{b}\ge q^*)]\mathbb{P}(s\le p^*)\\
        & = \mathbb{E}[(\max\{q^*,\underline{b}\}-p^*)\mathbb{I}(\bar{b}\ge q^*)]\left(\mathbb{P}(s\le p_{k^*})-\mathbb{P}(s\le p^*)\right)-\mathbb{P}(q^*\le \bar{b}< q_{j^*}) +(p^*-p_{k^*})\\
        & \ge 0-\frac{1}{K}-2\sqrt{\frac{\ln(\frac{2}{\delta})}{2T_0}}-\frac{1}{K}\\
	& = -\frac{2}{K}-2\sqrt{\frac{\ln(\frac{2}{\delta})}{2T_0}},
	\end{align*}
	where we used $p^*\le p_{k^*}$ and $q^*\le q_{j^*}$, and in the last inequality \Cref{eq:positiveRev}.
    
    \paragraph{$\mathcal{P}$ has density bounded by $M$}
    Here, we simply observe that 
    \begin{align*}
	&\rev(p_{k^*},q_{j^*})-\rev(p^*,q^*)\\
        &\hspace{0.5cm} = \mathbb{E}[(\max\{q_{j^*},\underline{b}\}-p_{k^*})\mathbb{I}(\bar{b}\ge q_{j^*})\mathbb{I}(s\le p_{k^*})] -\mathbb{E}[(\max\{q^*,\underline{b}\}-p^*)\mathbb{I}(\bar{b}\ge q^*)\mathbb{I}(s\le p^*)]\\
        &\hspace{0.5cm} \ge - \mathbb{P}(q^* \le \bar{b}_t \le q_{j^*}) - \mathbb{P}(p^* \le s_t \le p_{k^*})\\
        &\hspace{0.5cm}  \ge -\frac{2M}{K}.
	\end{align*}

\end{proof}

\revdisc*
\begin{proof}
    In the following, to make the notation more agile we let $p=p_{k^*}$ and $q=q_{j^*}$.
    If $p\ge q$ then, we observe that
	\begin{align*}
		\mathbb{E}[\left(\bar{b}-p\right)\mathbb{I}(\bar{b}\ge q)\mathbb{I}(s\le p)]
		& \le \mathbb{E}[\left(\bar{b}-p\right)\mathbb{I}(p >\bar{b} \ge q )\mathbb{I}(s\le p)] + \mathbb{E}[\left(\bar{b}-p\right)\mathbb{I}(\bar{b} \ge p )\mathbb{I}(s\le p)]\\
		& \le E[ \left(\bar{b}-p\right)\mathbb{I}(\bar{b} \ge p )\mathbb{I}(s\le p)]\\
		& \le \sum_{j=0}^{\log{T}} E\left[ \left(\bar{b}-p\right) \mathbb{I}\left(p+\frac{1}{2^{j}}\le \bar{b} \le p+\frac{1}{2^{j-1} }\right)\mathbb{I}(s\le p)\right] + \frac{1}{T}\\
		& \le 2\sum_{j=0         }^{\log{T} } \mathbb{E}\left[ \frac{1}{2^j}\mathbb{I}\left(q+\frac{1}{2^{j}}\le \bar{b}_t \le q+\frac{1}{2^{j-1} }\right)\mathbb{I}(s_t\le p) \right]+ \frac{1}{T}\\
		& \le 3\log(T) \max_{(p',q')\in F_K^+}\rev(p',q') + \frac{1}{T},
	\end{align*}
	where the last inequality holds by observing that $\mathbb{I}(p+\frac{1}{2^{j}}\le \bar{b} \le p+\frac{1}{2^{j-1} })=1$ implies that posting price $(p,p+2^{-j})$, which belongs to $F_K^+$ by construction, the obtained revenue is $2^{-j}$.

    We apply an analogous reasoning for
	\begin{align*}
		\mathbb{E}[\left(\max\{q,\underline{b}\}-s\right)&\mathbb{I}(\bar{b}\ge q)\mathbb{I}(s\le p)]\\
		& \le \mathbb{E}[\left(\max\{q,\underline{b}\}-s\right) \mathbb{I}(\bar{b}\ge q)\mathbb{I}(\max\{q,\underline{b}\}<s\le p)] \\
		& \hspace{3cm}+  \mathbb{E}\left[\left(\max\{q,\underline{b}\}-s\right) \mathbb{I}(\bar{b}\ge q)\mathbb{I}(s\le \max\{q,\underline{b}\})\right]\\
		& \le\sum_{j=0}^{\log{T} }\mathbb{E} \left[ \left(\max\{q,\underline{b}\} - s\right) \mathbb{I}\left(\max\{q,\underline{b}\}-\frac{1}{2^{j-1}}\le s \le q-\frac{1}{2^{j}}\right)\right] + \frac{1}{T}\\
		& \le 2 \sum_{j=0}^{\log{T} } \mathbb{E}\left[ 2^{-j} \mathbb{I}\left(\max\{q,\underline{b}\}-\frac{1}{2^{j-1}}\le s \le \max\{q,\underline{b}\}-\frac{1}{2^{j}}\right)\right] + \frac{1}{T}\\
		& \le 3\log(T) \max_{(p',q')\in F_K^-}\rev(p',q') + \frac{1}{T},
	\end{align*}
    where the last inequality holds since if $\mathbb{I}(\max\{q,\underline{b}_t\}-\frac{1}{2^{j-1}}\le s_t \le \max\{q,\underline{b}_t\}-\frac{1}{2^{j}})=1$ the revenue achieved by posting prices $(\max\{q,\underline{b}_t\}-2^{-j},q)$, which belongs to $F_K^-$, is at least $2^{-j}$.
    
     Similarly, if it holds instead that $p<q$, then 
    \begin{align*}
		\mathbb{E}\left[\left(\bar{b}-p\right)\mathbb{I}(\bar{b}\ge q)\mathbb{I}(s\le p)\right]& \le  \mathbb{E}\left[\left(\bar{b}-p\right)\mathbb{I}(\bar{b} \ge p )\mathbb{I}(s\le p)\right]\\
		& \le \sum_{j=0}^{\log{T}} E\left[ \left(\bar{b}-p\right) \mathbb{I}\left(p+\frac{1}{2^{j}}\le \bar{b} \le p+\frac{1}{2^{j-1} }\right)\mathbb{I}(s\le p)\right] + \frac{1}{T}\\
		& \le 2\sum_{j=0}^{\log{T} } \mathbb{E}\left[ \frac{1}{2^j}\mathbb{I}\left(q+\frac{1}{2^{j}}\le \bar{b}_t \le q+\frac{1}{2^{j-1} }\right)\mathbb{I}(s_t\le p) \right]+ \frac{1}{T}\\
		& \le 3\log(T) \max_{(p',q')\in F_K^+}\rev(p',q') + \frac{1}{T},
	\end{align*}
    as $p< q$ implies $\left(\bar{b}-p\right)\mathbb{I}(p\le\bar{b} < q )\mathbb{I}(s\le p)\ge 0$, and
    \begin{align*}
		\mathbb{E}[\left(\max\{q,\underline{b}\}-s\right)&\mathbb{I}(\bar{b}\ge q)\mathbb{I}(s\le p)]\\
		&\le  \mathbb{E}[\left(\max\{q,\underline{b}\}-s\right) \mathbb{I}(\bar{b}\ge q)\mathbb{I}(s\le \max\{q,\underline{b}\})]\\
		& \le\sum_{j=0}^{\log{T} }\mathbb{E} \left[ \left(\max\{q,\underline{b}\} - s\right) \mathbb{I}\left(\max\{q,\underline{b}\}-\frac{1}{2^{j-1}}\le s \le q-\frac{1}{2^{j}}\right)\right] + \frac{1}{T}\\
		& \le 2 \sum_{j=0}^{\log{T} } \mathbb{E}\left[ 2^{-j} \mathbb{I}\left(\max\{q,\underline{b}\}-\frac{1}{2^{j-1}}\le s \le \max\{q,\underline{b}\}-\frac{1}{2^{j}}\right)\right] + \frac{1}{T}\\
		& \le 3\log(T) \max_{(p',q')\in F_K^-}\rev(p',q') + \frac{1}{T}.
	\end{align*}

    Combining these results we get;
    \begin{subequations}
        \begin{align}
	   \gft(p,q)& = \mathbb{E}[\left(\bar{b}-s\right)\mathbb{I}(\bar{b}\ge q)\mathbb{I}(s\le p)] \nonumber\\
		&= \mathbb{E}[\left(\bar{b}-p\right)\mathbb{I}(\bar{b}\ge q)\mathbb{I}(s\le p)] + \mathbb{E} [\left(\max\{q,\underline{b}\}-s\right)\mathbb{I}(\bar{b}\ge q)\mathbb{I}(s\le p)] \nonumber \\
		&\hspace{4.35cm}+ \mathbb{E}[p-\max\{q,\underline{b}\})\mathbb{I}(\bar{b}\ge q)\mathbb{I}(s\le p)] \nonumber\\
		& \le \max_{(p',q')\in F_K} 6\log(T) \rev(p',q') + \frac{2}{T}- \rev(p,q) \label{eq: revmax eq2'}\\
		& \le  6\log(T)\max_{(p',q')\in F_K} \rev(p',q') + 2/T +  \zeta   \label{eq: revmax eq3'},
	\end{align}
    \end{subequations}

	where Inequality \eqref{eq: revmax eq2'} holds by definition of Revenue and the two previous results and Inequality~\eqref{eq: revmax eq3'} holds by the definition of $\zeta$.
  \end{proof}

\maxrevhedge*
\begin{proof}
    We can use UCB1 on the set of arms 
    $F_K$ to implement the \texttt{maxREV} algorithm.
    This guarantees that with probability at least $1-\delta$.
    \[ \sum_{t\in [\tau]}\rev(p_t,q_t)\ge \tau \rev(p^*_j,q^*_j) - \tilde O(T^{2/3}).  \]
    Moreover, an Azuma-Hoeffding inequality implies:
    \[ \sum_{t\in [\tau]} \rev_t(p_t,q_t)\ge  \sum_{t\in [\tau]}\rev(p_t,q_t) - \tilde O(\sqrt{T}),\]
    with probability at least $1-\delta$.

    Hence, taking an union bound and applying \cref{lemma: rev opt disc}, we get:
	\begin{align*}
        \tau \gft(p^*_j,q^*_j)&\le 6\log(T) \tau\max_{(p',q')\in F_K} \rev(p',q') + 2 + \tau \zeta\\ 
		& \le  6 \log(T) \sum_{t\in [\tau]} \rev(p_t,q_t) + \tilde O(T^{2/3}) + \tau \zeta \\
        & \le  6 \log(T) \sum_{t\in [\tau]} \rev_t(p_t,q_t) + \tilde O(T^{2/3}) + \tau \zeta.
	\end{align*}
   
\end{proof}

\section{Proofs Omitted from \Cref{sec:SBB}}


\sbbtheo*
\begin{proof}
In this proof we will use $q^*$ to denote the optimum price, \emph{i.e.} $q^*\in \argmax_{q\in[0,1]}\gft(q)$, and $q_{k^*}$ to denote the optimum price over the grid $\mathcal{B}$, \emph{i.e.} $q_{k^*}\in \argmax_{q\in \mathcal{B}}\gft(q)$.
	First we can decompose the regret in 5 different components representing respectively the regret induced by the approximation to the grid $\mathcal{B}$ ($R_K$), the regret generated in the pure exploration phase ($(1)$), the regret induced by the error of estimation of the actual $GFT$ (components$(2)$ and $(4)$) and the regret induced by the algorithm over the optimistic estimate of $GFT$ ($(3)$).
   
	Formally:
	\begin{align*} 
		R_T & = T\cdot \gft(q^*)-\sum_{t=1}^T\gft(q_t)\\
		& = \underbrace{\sum_{t=1}^{2T_0}(\gft(q^*)-\gft(q_t))}_{(1)} +  \underbrace{\sum_{t=2T_0+1}^{T}\left(\gft(q^*)-\gft(q_{k^*})\right)}_{R_K} \\
		&+ 
		 \underbrace{\sum_{t=2T_0+1}^{T}(\gft(q_{k^*})-\overline{\gft}_t(q_{k^*}))}_{(2)}\\
		& + \underbrace{\sum_{t=2T_0+1}^{T}(\overline{\gft}_t(q_{k^*})-\overline{\gft}_t(q_t))}_{(3)} + \underbrace{\sum_{t=2T_0+1}^{T}(\overline{\gft}_t(q_t)-\gft(q_t))}_{(4)}\\
		& = R_K + (1)+(2)+(3)+(4).
	\end{align*}
	In the following, we proceed to bound each component separately. 
	
	\paragraph{Bound $(1)$}:  It can be bounded trivially by $2T_0$ .
	
	\paragraph{Bound $R_K$:}
	With probability at least $1-\delta$ by Lemma~\ref{lemma: gridsbb} 
	\begin{equation*}
		R_K \le \BigOL{\frac{T}{K}+ \frac{T}{\sqrt{T_0}}}.
	\end{equation*}

	\paragraph{Bound $(2)$:}
	To bound $(2)$ we first prove that by construction with probability at least $1-2\delta$ the estimated game for trade $\overline{\gft}_t(q)$ is greater or equal than $\gft(q)$ for all episode $t\in [2T_0+1,\ldots,T]$ and for all $q\in \mathcal{B}$.
	
	First observe that $\widehat{\gft}_1$ is such that 
	\begin{align*}
		\mathbb{E}[\widehat{\gft}_1(q_k)]& =  \mathbb{E}\left[ \frac{1}{T_0}\sum_{t=T_0+1}^{2T_0}\mathbb{I}(s_t\le U_t \le \max\{q_k,\underline{b}_t\})\mathbb{I}(\bar{b}_t \ge q_k)\right]\\
		& = \gft_1(q_k)
	\end{align*}
	 using Lemma~\ref{lemma: gft decomposition}.
	
	In addition, it holds that  $\lvert\mathbb{I}(s_t\le U_t \le \max\{q_k,\underline{b}_t\})\mathbb{I}(\bar{b}_t \ge q_k)\rvert\le 1$.
	
	Thus by Hoeffding inequality with probability at least $1-\delta$ for all $q_k\in \mathcal{B}$:
	\begin{equation}
		\label{eq:AH gft1}
		\bigg\lvert \widehat{\gft}_1(q_k) -  \gft_1(q_k)\bigg\rvert \le 2\sqrt{\frac{\ln(\frac{2K}{\delta})}{2T_0}},
	\end{equation}
	and therefore with probability at least $1-\delta$:
	\begin{equation*}
		\overline{\gft}_1(q_k)= \widehat{\gft}_1(q_k) + 2\sqrt{\frac{\ln(\frac{2K}{\delta})}{2T_0}} \ge \gft_1(q_k).
	\end{equation*}
	For what concern $\gft_2$ observe that $\widehat{\gft}_{t,2}$ is, for all $q_k$ in $\mathcal{B}$ and for all $t\in [2T_0+1,T]$, defined as
	\begin{equation*}
		\widehat{\gft}_{t,2}(q_k)= \frac{1}{n_t(q_k)}\sum_{t=2T_0+1}^t (-\max\{q_k,\underline{b}_t\}+\bar{b_t})\mathbb{I}(\bar{b}_t\ge q_k)\mathbb{I}(s_t\le \max\{\underline{b}_t,q_k\}).
	\end{equation*}
	Thus, by applying also in this case the Hoeffding inequality for all $q_k\in \mathcal{B}$, $t\in [2T_0+1,\ldots,T]$ with probability at least $1-\delta$ it holds:
	\begin{equation}
		\label{eq:AH gft2}
		\bigg\lvert \widehat{\gft}_{2,t}(q_k) -  \gft_2(q_k)\bigg\rvert \le 2\sqrt{\frac{\ln(\frac{2KT}{\delta})}{2n_t(q_k)}},
	\end{equation}
	and therefore
	\begin{equation*}
		\overline{\gft}_{t,2}(q_k)= \widehat{\gft}_{2,t}(q_k) + \sqrt{\frac{\ln(\frac{2KT}{\delta})}{2n_t(q_k)}} \ge \gft_2(q_k).
	\end{equation*}
	Thus, with probability at least $1-2\delta$ 
	\begin{align*}
		(2)\le 0.
	\end{align*}
	\paragraph{Bound $(3)$:}
	By construction of the algorithm $(3)\le 0$.
	\paragraph{Bound $(4)$:} To bound $(4)$ we first observe that by Inequality \eqref{eq:AH gft1} with probability at least $1-\delta$ for all $q\in \mathcal{B}$
	\begin{equation*}
		\overline{\gft}_1(q_k)-\gft_1(q_k)\le 2 \left(2\sqrt{\frac{\ln(\frac{2K}{\delta})}{2T_0}}\right),
	\end{equation*}
	and by Inequality \eqref{eq:AH gft2} it holds with probability at least $1-\delta$ for all $q_k\in \mathcal{B}$ and for all $t \in [2T_0+1,T]$ 
	\begin{equation*}
		\overline{\gft}_{t,2}(q_k)-\gft_2(q_k)\le 2 \left(2\sqrt{\frac{\ln(\frac{2KT}{\delta})}{2n_t(q_k)}}\right),
	\end{equation*}
    where $n_t(q_k)$ is the counter of in how many episodes from $2T_0+1$ to $t$ the algorithm played the price $q_k$.
	Therefore, with probability at least $1-2\delta$
	\begin{align*}
		(4)& = \sum_{t=2T_0+1}^{T}(\overline{\gft}_t(q_t)-\gft(q_t))\\
		& \le 4(T-2T_0)\sqrt{\frac{\ln(\frac{4KT}{\delta})}{2T_0}} + 4\sum_{t=2T_0+1}^T\sqrt{\frac{\ln(\frac{4KT}{\delta})}{n_t(q_t)}}\\
		& = \BigOL{\frac{T}{\sqrt{T_0}}+\sqrt{KT}}.
	\end{align*}
	
	To conclude the proof we join all the results and we observe that with probability at least $1-3\delta$
	\begin{align*}
		R_t &= R_K + (1)+(2)+(3)+(4)\\
		& \le \BigOL{\frac{T}{K}+\frac{T}{\sqrt{T_0}}}+ \BigOL{T_0} + \BigOL{\frac{T}{\sqrt{T_0}}+\sqrt{KT}}\\
		& \le \BigOL{\frac{T}{K}+\frac{T}{\sqrt{T_0}}+ T_0+ \sqrt{KT}}\\
		& = \BigOL{T^{2/3}},
	\end{align*}
	where the last equality holds if we fix $K=\BigOL{T^\frac{1}{3}}$, $T_0=\BigOL{T^\frac{2}{3}}$.
	
\end{proof}

\section{Proofs Omitted from \Cref{sec:feedback}}
\theoremGMAB*
\begin{proof}
The proof can be divided in two parts; in the first we prove that the optimal arm belongs with high probability to the safe set at each episode and in the second we actually bound the regret between the rewards of the arms chosen and an arbitrary arm in the safe set.

    \paragraph{Optimal arm belongs to $\bigcap_{t\in[T]}\mathcal{S}_t$}
    First we prove the following two auxiliary results:
     \begin{enumerate}
        \item  with probability at least $1-\delta$ for all $t\in [T]$ and for all $a_{i,j}\in \mathcal{A}$ 
    \begin{equation*}
        r(a_{i,j}) \in [\underline{r}_t(a_{i,j}),\overline{r}_t(a_{i,j})].
    \end{equation*}
    Since $\mathbb{E}[\widehat{r}_t(a_{i,j})]=\mathbb{E}[\frac{1}{n_t(a_{i,j})}\sum_{\ell=1}^t r_\ell(a_{i,j})\mathbb{I}(a_{i,j}\in Obs_\ell)]=\frac{n_t(a_{i,j})r(a_{i,j})}{n_t(a_{i,j})}= r(a_{i,j})$, where $n_t(a_{i,j})$ is the number of times the reward of arm $a_{i,j}$ is observed up to episode $t$ and where $\mathbb{I}(a_{i,j}\in Obs_\ell)=1$ means that at episode $\ell$, as consequence of playing arm $a_{i_\ell,j_\ell} $ the reward of arm $a_{i,j}$ is also disclosed, and since $|r_\ell(a_{i,j})\mathbb{I}(a_{i,j}\in Obs_\ell)|\le 1$ for all $\ell\in [T]$, then by Hoeffding inequality with probability at least $1-\delta$
    \begin{equation*}
       \bigg\lvert \widehat{r}_t(a_{i,j})-r(a_{i,j})\bigg \rvert \le 2\sqrt{\frac{\ln(\frac{2nmT}{\delta})}{2n_t(a_{i,j})}}.
    \end{equation*}
   
    \item With probability at least $1-\delta$ for all $t\in [T]$ and for all $(a_{i,j})\in \mathcal{A}$ 
    \begin{equation*}
        c(a_{i,j}) \in [\underline{c}_t(a_{i,j}),\overline{c}_t(a_{i,j})].
    \end{equation*}
    We prove this point in an analogous way. By means of Hoeffding inequality, recalling the definition of $\widehat{c}_t(a_{i,j})=\frac{1}{n_t(a_{i,j})}\sum_{\ell=1}^t c_\ell(a_{i,j})\mathbb{I}(a_{i,j}\in Obs_\ell)]$, it holds with probability at least $1-\delta$ for all $t\in [T]$ and for all $a_{i,j}\in \mathcal{A}$
    \begin{equation*}
       \bigg\lvert \widehat{c}_t(a_{i,j})-c(a_{i,j})\bigg \rvert \le 2\sqrt{\frac{\ln(\frac{2nmT}{\delta})}{2n_t(a_{i,j})}}.
    \end{equation*}
    \end{enumerate}
   Finally, we are ready to prove that $a_{i^*,j^*}\in \mathcal{S}_t$ for all $t\in[T]$. Indeed there are two possible mechanism of arms elimination.
   
   The first is straightforward and consists in eliminating all arms for which it exists a $t$ such that $\underline{c}_t(a_{i,j})\ge 0$, which, with probability at least $1-\delta$, does not eliminate $a_{i^*,j^*}$ since by definition $c(a_{i^*,j^*})\le 0$.
   
   The second method is less obvious and eliminate only arms $a_{i,j'} $ for which it exist $t\in[T]$ and $a_{i,j}\in \mathcal{S}_t$ such that $\overline{r}(a_{i,j'})\le \underline{r}(a_{i,j})$ and $j'\le j$. We distinguish two possible cases, either $c(a_{i,j})\le 0$ or $c(a_{i,j})>0$.
   If $c(a_{i,j})\le 0$ then $a_{i,j}$ is feasible and $a_{i,j'}$ is with high probability sub-optimal, indeed with probability at least $1-\delta$ it holds $r(a_{i,j'})\le\overline{r}(a_{i,j'})\le \underline{r}(a_{i,j}) \le r(a_{i,j}) $.
   Otherwise, if $c(a_{i,j})> 0$ then $c(a_{i,j'})\ge c(a_{i,j})> 0$ and $a_{i,j'}$ is not feasible.
   \paragraph{Bounding regret}
   The regret can be decomposed in four parts:
   \begin{align*}
        \tau\cdot r(a_{i^*,j^*})-\sum_{t=1}^\tau r(a_{i_t,j_t})& = \tau\cdot r(a_{i^*,j^*})-\sum_{t=1}^\tau \overline{r}_{t-1}(a_{i^*,j^*})\\
        & + \sum_{t=1}^\tau (\overline{r}_{t-1}(a_{i^*,j^*})-\overline{r}_{t-1}(a_{i_t,\bar{j}_t}))\\
        & +  \sum_{t=1}^\tau (\overline{r}_{t-1}(a_{i_t,\bar{j}_t})-\overline{r}_{t-1}(a_{i_t,j_t}))\\
        &+  \sum_{t=1}^\tau (\overline{r}_{t-1}(a_{i_t,j_t})-r(a_{i_t,j_t})).
   \end{align*}
   The first term is with probability at least $1-\delta$ upper bounded by $0$ by construction of $\overline{r}$.
   The second term is also upper bounded by $0$ with probability at least $1-2\delta$ as $a_{i_t,\bar{j}_t}$ is the arm that maximize $\overline{r}_{t-1}$ in $\mathcal{S}_{t-1}$ and as we proved above $a_{i^*,j^*}$ is in $\mathcal{S}_{t-1}$ with probability at least $1-2\delta$.
   The third term can be further decomposed as 
   \begin{align*}
       \sum_{t=1}^\tau (\overline{r}_{t-1}(a_{i_t,\bar{j}_t})-\overline{r}_{t-1}(a_{i_t,j_t}))& = \sum_{t=1}^\tau \left(\underline{r}_{t-1}(a_{i_t,\bar{j}_t})+ 4\sqrt{\frac{\ln(\frac{2nmT}{\delta})}{2n_{t-1}(a_{i_t,\bar{j}_t})}} -\overline{r}_{t-1}(a_{i_t,j_t})\right)\\
       & \le \sum_{t=1}^\tau 4\sqrt{\frac{\ln(\frac{2nmT}{\delta})}{2n_{t-1}(a_{i_t,\bar{j}_t})}}\\
       & = \sum_{t=1}^\tau 4\sqrt{\frac{\ln(\frac{2nmT}{\delta})}{2n_{t-1}(i_t)}}\\
       & = \BigOL{\sqrt{\tau n \log(nmT/\delta)}},
   \end{align*}
   since by construction $n_{t-1}(a_{i_t,\bar{j}_t})=n_{t-1}(a_{i_t,j'})$ for all $a_{i_t,j'} \in \mathcal{S}_{t-1}$, and we indicated that counter as simply $n_{t-1}(i_t)$.
   
   The fourth term can be treated similarly, hence with probability at least $1-\delta$
   \begin{align*}
       \sum_{t=1}^\tau (\overline{r}_{t-1}(a_{i_t,j_t})-r(a_{i_t,j_t}))& = \sum_{t=1}^\tau \left(\underline{r}_{t-1}(a_{i_t,j_t})+ 4\sqrt{\frac{\ln(\frac{2nmT}{\delta})}{2n_{t-1}(a_{i_t,j_t})}} -r(a_{i_t,j_t})\right)\\
       & \le \sum_{t=1}^\tau \left(4\sqrt{\frac{\ln(\frac{2nmT}{\delta})}{2n_{t-1}(a_{i_t,j_t})}}\right)\\
       & = \sum_{t=1}^\tau \left(4\sqrt{\frac{\ln(\frac{2nmT}{\delta})}{2n_{t-1}(i_t)}}\right)\\
       &= \BigOL{\sqrt{\tau n \log(nmT/\delta)}}.
   \end{align*}
   Thus, joining all the bounds we get that with probability at least $1-2\delta$ for all $\tau\in [T]$ it holds
   \begin{equation*}
       \tau\cdot r(a_{i^*,j^*})-\sum_{t=1}^\tau r(a_{i_t,j_t}) \le \BigOL{\sqrt{\tau n \log(nmT/\delta)}}.
   \end{equation*}
   \paragraph{Costs constraints}
   To conclude the proof we observe that with probability at least $1-\delta$ for all $\tau\in [T]$
   \begin{align*}
       \sum_{t=1}^\tau [c(a_{i_t,j_t})]^+& \le \sum_{t=1}^\tau [\overline{c}_{t-1}(a_{i_t,j_t})]^+\\
       & =  \sum_{t=1}^\tau \left[\left(\underline{c}_{t-1}(a_{i_t,j_t}) + 4\sqrt{\frac{\ln(\frac{2nmT}{\delta})}{2n_{t-1}(i_t)}}\right)\right]^+\\
       & \le \sum_{t=1}^\tau 4\sqrt{\frac{\ln(\frac{2nmT}{\delta})}{2n_{t-1}(i_t)}}\\
       & = \BigOL{\sqrt{\tau n \log(nmT/\delta)}}.
   \end{align*}
   
\end{proof}

\section{Proofs Omitted from  \Cref{sec:GBBalgorithm}}\label{app:GBB}

\gbbtheo*
\begin{proof}
	We first prove the statement relative to the regret.
	\begin{align*}
		R_T & = T \cdot \gft(p^*,q^*)-\sum_{t=1}^T\mathbb{E}[\gft(p_t,q_t)]\\
		& = \underbrace{T_0\cdot\gft(p^*,q^*)-\sum_{t=1}^{T_0}\mathbb{E}[\gft(p_t,q_t)]}_{(1)} + \underbrace{(T-T_0)\gft(p^*,q^*)-(T-T_0)\gft(p_{k^*},q_{j^*})}_{R_K} \\
		&+ \underbrace{(\tau-T_0) \gft(p_{k^*},q_{j^*})- \sum_{t=T_0+1}^\tau \mathbb{E}[\gft(p_t,q_t)]}_{R_{rev}} \\
		&+ \underbrace{T_0\cdot\gft(p_{k^*},q_{j^*})-\sum_{t=\tau+1}^{\tau+T_0}\mathbb{E}[\gft(p_t,q_t)]}_{(2)}\\
		& + \underbrace{(T-\tau-T_0){\gft}_t(p_{k^*},q_{j^*})-\sum_{t=\tau+T_0+1}^T\mathbb{E}\left[{\gft}_t(p_t,q_t)\right]}_{(3)} \\
		& = R_K + R_{rev}+(1)+ (2)+(3).
	\end{align*}
	\paragraph{Bound $R_K$:}
	With probability at least $1-\delta$ by Lemma~\ref{lemma: gridgbb} if $\mathcal{P}$ is bounded and/or sellers and buyers valuations are independent
	\begin{equation*}
		R_K \le \BigOL{\frac{T}{K}+\frac{T}{\sqrt{T_0}}}= \BigOL{T^{2/3}}.
	\end{equation*}
	\paragraph{Bound $R_{rev}$:}
	By Lemma~\ref{lemma: maxrevhedge} with probability at least $1-2\delta$
	\begin{equation*}
		R_{\rev}\le \BigOL{\beta + \tau\zeta+ T^{2/3}}= \BigOL{T^{2/3}}.
	\end{equation*}
	\paragraph{Bound $(1)$:}
	$(1)$ can be bounded trivially by $T_0$
    \[(1) \le T_0 = \BigOL{T^{2/3}}.\]
	\paragraph{Bound $(2)$:}
	$(2)$ can be also bounded trivially by $T_0$
    \[(2) \le T_0 = \BigOL{T^{2/3}}.\]
	\paragraph{Bound $(3)$:}
	
	If the distributions of sellers and buyers are independent and/or $\mathcal{P}$ is bounded it exists a choice of $\bar\zeta= \BigOL{T^{-1/3}}$ for which with probability at least $1-\delta$ $\rev(q_{k^*},q_{j^*})\ge -\bar\zeta$ by Lemma~\ref{lemma: -revk}. If the algorithm is initiated with that $\bar\zeta$ then it holds with probability at least $1-\delta$ that $c(p_{k^*},q_{j^*})=\mathbb{E}[c_t(p_{k^*},q_{j^*})]=\mathbb{E}[\frac{1}{\bar\zeta+1}(-\rev_t(p_{k^*},q_{j^*})-\bar\zeta)]=\frac{1}{\bar\zeta+1}(-\rev(p_{k^*},q_{j^*})-\bar\zeta)\le 0$.
    Hence, we can apply \cref{theoremGMAB} to bound $(3)$ and with probability at least $1-3\delta$ it holds 
    \begin{align*}
        (3) & = (T-\tau-T_0){\gft}(p_{k^*},q_{j^*})-\sum_{t=\tau+T_0+1}^T{\gft}(p_t,q_t)\\
        & =  \sum_{t=\tau+T_0+1}^T \left({\gft}_1(p_{k^*},q_{j^*})-\gft_2(p_{k^*},q_{j^*})\right)-\sum_{t=\tau+T_0+1}^T\left({\gft}_1(p_t,q_t)-{\gft}_2(p_t,q_t)\right)\\
        & = \sum_{t=\tau+T_0+1}^T \left(\gft_1(p_{k^*},q_{j^*})- \overline{\gft}_1(p_{k^*},q_{j^*})\right) + \sum_{t=\tau+T_0+1}^T \left(\overline{\gft}_1(p_t,q_t)-\gft_1(p_t,q_t) \right)\\
        & \hspace{2cm} + 2\sum_{t=\tau+T_0+1}^T\left(\frac{1}{2}(\overline{\gft}_1(p_{k^*},q_{j^*})+ \gft_{2}(p_{k^*},q_{j^*}))-\frac{1}{2}(\overline{\gft}_1(p_t,q_t)+ \gft_{2}(p_t,q_t))\right)\\
        & \le 8\sqrt{\frac{\ln(2|\mathcal{B}|/\delta)}{2T_0}}(T-T_0-\tau) + 2\sum_{t=\tau+T_0+1}^T \left(r(p_{k^*},q_{j^*})-r(p_t,q_t)\right) \\
        & \le \BigOL{\frac{T}{\sqrt{T_0}}} + \BigOL{\sqrt{KT\log(K/\delta)}}\\
        & = \BigOL{T^{2/3}},
    \end{align*}
    where the first inequality holds by the application of Hoeffding inequality  for which with probability at least $1-\delta$ for all $(p,q)\in \mathcal{B}$
    \begin{align*}
        |\gft_1(p,q)-\widehat{\gft}_1(p,q)|& = \mathbb{E}[\gft_1(p,q)-\widehat{\gft}_1(p,q)]\\
        &\hspace{1cm}+ \bigg|(\gft_1(p,q)-\widehat{\gft}_1(p,q))-\mathbb{E}[\gft_1(p,q)-\widehat{\gft}_1(p,q)]\bigg|\\
        & \le \mathbb{E}\left[\frac{1}{T_0}\sum_{t=\tau+1}^{T_0} \mathbb{I}\left(s_t\le U_t \le p\right)\mathbb{I}(\overline{b}_t\ge q)-\mathbb{P}(s_t\le U_t \le p,\overline{b}_t\ge q)\right] \\
        &\hspace{1cm} + 4\sqrt{\frac{\ln(\frac{2|\mathcal{B}|}{\delta})}{2T_0}}\\
        & = 4\sqrt{\frac{\ln(\frac{2|\mathcal{B}|}{\delta})}{2T_0}},
    \end{align*}
    and therefore 
    \[\gft_1(p,q) \in \left[\underline{\gft}_1,\overline{\gft}_1\right]=\left[\widehat{\gft}_1-4\sqrt{\frac{\ln(\frac{2|\mathcal{B}|}{\delta})}{2T_0}},\widehat{\gft}_1+4\sqrt{\frac{\ln(\frac{2|\mathcal{B}|}{\delta})}{2T_0}}\right].\]

	We prove now the result relative to the revenue of the algorithm.
	First observe that applying Hoeffding inequality with probability at least $1-\delta$
    \[\sum_{t=\tau+T_0+1}^T \rev_t(p_t,q_t) \ge \sum_{t=\tau+T_0+1}^T \rev(p_t,q_t) - 2\sqrt{\frac{\ln(\frac{2}{\delta})}{2}(T-T_0-\tau)}.\]
    With probability at least $1-3\delta$ using 
    \Cref{theoremGMAB} and \Cref{lemma: -revk}

		\begin{align*}
        \sum_{t=1}^T\rev_t(p_t,q_t)& = \sum_{t=1}^{T_0}\rev_t(p_t,q_t)  + \sum_{t=T_0+1}^{\tau}\rev_t(p_t,q_t) +  \sum_{t=\tau+1}^{\tau+T_0}\rev_t(p_t,q_t) + \sum_{t=\tau+T_0+1}^{T}\rev_t(p_t,q_t)\\
		& \ge -2T_0 + \beta + \sum_{t=\tau+T_0+1}^{T}\rev_t(p_t,q_t)\\
        & \ge -2T_0 + \beta + \sum_{t=\tau+T_0+1}^{T}\rev(p_t,q_t) - 2\sqrt{\frac{\ln(\frac{2}{\delta})}{2}(T-T_0-\tau)}\\
        & = -2T_0 + \beta  -\BigOL{\sqrt{T}} - (1+\bar\zeta)\sum_{t=\tau+T_0+1}^{T}\frac{1}{1+\bar{\zeta}}(-\rev(p_t,q_t)-\bar\zeta) + (T-\tau-T_0)\bar\zeta\\
        & \ge-2T_0 + \beta  -\BigOL{\sqrt{T}} - 2 \sum_{t=\tau+T_0+1}^{T} c(p_t,q_t) + \BigOL{T\bar\zeta}\\
		& \ge  \beta - \BigOL{T_0 + T\bar\zeta +\sqrt{KT}\log(K/\delta)} \\
        & \ge 0,
	\end{align*}
	where the last holds for some choice of $\beta=\BigOL{T^{2/3}}$.
	
\end{proof}

\end{document}